\documentclass[pdflatex,sn-mathphys-ay]{sn-jnl}% Math and Physical Sciences Numbered Reference Style

\usepackage{url}
\usepackage{indentfirst}
\usepackage{graphicx}%
\usepackage{multirow}%
\usepackage{amsmath,amssymb,amsfonts}%
\usepackage{amsthm}%
\usepackage{mathrsfs}%
\usepackage[title]{appendix}%
\usepackage{xcolor}%
\usepackage{textcomp}%
\usepackage{manyfoot}%
\usepackage{booktabs}%
\usepackage{algorithm}%
\usepackage{algorithmicx}%
\usepackage{algpseudocode}%
\usepackage{listings}%
\usepackage{lmodern}%
\usepackage{orcidlink}
\usepackage{makecell}
\usepackage{soul}

\allowdisplaybreaks

\theoremstyle{plain}
\newtheorem{theorem}{Theorem}
\newtheorem{proposition}[theorem]{Proposition}
\newtheorem{conjecture}[theorem]{Conjecture}
\newtheorem{corollary}{Corollary}[theorem]

\theoremstyle{definition}
\newtheorem{definition}{Definition}

\usepackage{color}
\definecolor{darkred}{rgb}{0.9,0,0}
\definecolor{darkblue}{rgb}{0,0,0.75}
\definecolor{darkgreen}{rgb}{0.1,0.6,0.3}
\definecolor{darkred}{rgb}{0.6,0.3,0.1}

\newcommand{\yan}[1]{\textcolor{darkred}{[#1]}}
\newcommand{\arm}[1]{\textcolor{blue}{[#1]}}

\newcommand{\be}{\begin{equation}}
\newcommand{\ee}{\end{equation}}
\def\1{{\bf 1}}

\DeclareMathOperator{\sign}{sign}

\begin{document}

\title{Expected and minimal values of a universal tree balance index}

\author[1,*]{Veselin Manojlovi\'c\,\orcidlink{0000-0003-0620-2431}}
\author[2,*]{Armaan Ahmed\,\orcidlink{0000-0001-9619-1611}}
\author[3]{Yannick Viossat\,\orcidlink{0000-0003-1388-0599}}
\author[1,**]{Robert Noble\,\orcidlink{0000-0002-8057-4252}}
\affil[1]{\orgdiv{Department of Mathematics}, \orgname{City St George's, University of London}, \orgaddress{\street{Northampton Square}, \city{London}, \postcode{EC1V 0HB}, \country{United Kingdom}}}
\affil[2]{Department of Applied Math \& Statistics, Johns Hopkins University, Baltimore, USA}
\affil[3]{Ceremade, CNRS, Université Paris-Dauphine, Université PSL, Paris, France}
\affil[*]{These authors contributed equally}
\affil[**]{robert.noble@city.ac.uk}
	
\abstract{Although the analysis of rooted tree shape has wide-ranging applications, notions of tree balance have developed independently in different domains. In computer science, a balanced tree is one that enables efficient updating and retrieval of data, whereas in biology tree balance quantifies bias in evolutionary processes. The lack of a precise connection between these concepts has stymied the development of universal indices and general results. We recently introduced a new tree balance index, $J^1$, that, unlike prior indices popular among biologists, permits meaningful comparison of trees with arbitrary degree distributions and node sizes. Here we explain how our new index generalizes a concept that underlies the definition of the weight-balanced tree, an important type of self-balancing binary search tree. Our index thus unifies the tree balance concepts of biology and computer science. We provide new analytical results to support applications of this universal index. First, we quantify the accuracy of approximations to the expected values of $J^1$ under two important null models: the Yule process and the uniform model. Second, we investigate minimal values of our index. These results help establish $J^1$ as a universal, cross-disciplinary index of tree balance that generalizes and supersedes prior approaches.}

\maketitle

\section{Introduction}

Broadly speaking, the balance of a rooted tree is the extent to which its terminal nodes (leaves) are evenly distributed among its branches. Indices that quantify tree balance have important applications in systematic biology \citep{heard1992patterns, mooers_inferring_1997, purvis2002phylogeny}, mathematical oncology \citep{Scott2018a, noble2022spatial, feder2024detecting}, and other fields \citep{
leventhal2012inferring, colijn2014phylogenetic, chindelevitch2021network, barzilai2023signatures}. More than twenty conventional tree balance or imbalance indices have been defined \citep{fischer2023tree} yet all have important shortcomings \citep{lemant_robust_2022, noble_new_2023}. Some conventional tree balance indices are defined only for bifurcating trees; others are meaningful only if no nodes have outdegree one. All require adjustments to compensate for correlation with tree size or degree distribution, which hampers the comparison of dissimilar trees. Lack of universality importantly limits practical applications.

We recently introduced a universal tree balance index, denoted $J^1$, that has no such flaws. $J^1$ is defined for any rooted tree topology and accounts for arbitrary node sizes (which may, for example, correspond to the population sizes of biological types in an evolutionary tree) \citep{lemant_robust_2022}. We proved that $J^1$ is robust, in the sense that it is insensitive to small changes in node sizes and to the removal of small nodes. We further showed that this index -- which derives from the Shannon entropy -- both unites and generalizes the two most popular prior approaches to quantifying tree balance in biology.

Our new index has diverse potential applications.
Applied to evolutionary trees, we and others have shown that $J^1$ outperforms conventional tree balance indices as a summary statistic for comparing trees inferred from empirical data to computational model predictions \citep{noble2022spatial, feder2024detecting}. 
$J^1$ also provides promising new ways to infer speciation mode from phylogenies \citep{freitas2024patch} and to predict the outcome of cancer immunotherapy from tumour evolutionary trees \citep{bozic2024neoantigen}.
%We have further shown that $J^1$ can be extended to account for branch lengths as part of a new universal system of tree shape indices with several advantages over current approaches \citep{noble_new_2023}.

Given any new tree shape index, an important task is to obtain its expected and extremal values under standard tree-generating processes, which can then be used as null-model reference points  \citep{fischer2023tree}. In \citet{lemant_robust_2022} we obtained analytical approximations to the expected values of $J^1$ under the Yule process and the uniform model, and we tested their accuracy numerically for trees with up to 128 leaves. We also showed that the caterpillar tree does not always minimize $J^1$ among trees with only zero-sized internal nodes, uniform leaf sizes, and no nodes of outdegree one. Generalizing these results has remained an open problem.
	
The contributions of this paper are threefold. First, we further establish $J^1$ as a universal index of tree balance by identifying fundamental connections to classical results in computer science, related to Huffman coding and self-balancing tree data structures. Second, we prove the accuracy of our expected value approximations for the Yule process and the uniform model. For the Yule process, we derive a new, closer approximation that rapidly converges to the true expectation in the large-tree limit. Finally, we investigate the minimal values of $J^1$ in an important special case, obtaining a counter-intuitive result in the large tree limit.

\section{Results}
\subsection{Preliminary definitions}\label{defsec}

\begin{definition}[Rooted tree]
	A \textbf{rooted tree} $T$ is a connected acyclic graph in which one node is designated the root. Parent-child and ancestor-descendant relationships in a rooted tree are assigned along paths directed away from the root.
\end{definition}

\begin{definition}[Node size and tree magnitude, \citet{lemant_robust_2022}]
	We assign to each node $i$ a non-negative \textbf{node size}, $w_i$. The \textbf{magnitude} of a tree $T$ is the sum of its node sizes: 
 $$
 S(T) = \sum_{i \in V(T)} w_i,
 $$
 where $V(T)$ is the set of all nodes.
 To avoid confusion we will not refer to the size of a tree, which is conventionally defined as its node count.
\end{definition}

\begin{definition}(Leafy tree, \citet{lemant_robust_2022}) A \textbf{leafy tree} is one with only zero-sized internal nodes.
\end{definition}

\begin{definition}(Node depth and tree height) We define the \textbf{depth} of a node as the number of edges in the shortest path from the root to that node, and the \textbf{height} of a tree as its maximum node depth.
\end{definition}

\begin{definition}[Sackin index, \citet{sackin_good_1972}, external path length and internal path length]\label{sackin_defn}
	The \textbf{Sackin index} of rooted tree $T$ is the sum of its leaf depths: 
	\begin{equation}\label{sackindef}
		I_S(T) = \sum_{l\in L(T)} \nu(l),
	\end{equation}
	where $L(T)$ is the set of all leaves (terminal nodes) of $T$, and $\nu(l)$ is the depth of leaf $l$. In computer science this is known as the \textbf{external path length}, whereas the \textbf{internal path length} is the sum of all internal node depths \citep{knuth1997art_vol1} (p. 400).
\end{definition}

\begin{definition}[Generalized Sackin index, \citet{lemant_robust_2022} or weighted path length] The Sackin index can be generalized to account for arbitrary node sizes:
	\begin{equation}\label{gensackindef}
		I_{S,\text{gen}}(T) = \sum_{i \in V(T)} w_i \nu(i) = \sum_{i\in \tilde{V}(T)} S_i^*,
	\end{equation}
where $\tilde{V}(T)$ the set of all internal nodes (non-leaves) whose descendants are not all of zero size, and $S_i^*$ is the magnitude of the subtree rooted at node $i$, excluding $i$. If $T$ is a leafy tree in which all leaves have unit size then $I_{S,\text{gen}}(T) = I_S(T)$. In computer science, $I_{S,\text{gen}}(T)$ is known as the \textbf{weighted path length} \citep{knuth1997art_vol1} (p. 402), denoted $|T|$.
\end{definition}

\begin{definition}[Binary, bifurcating, full $m$-ary, and linear trees]
	A \textbf{binary tree} is a rooted tree in which no node has more than two children. A \textbf{full $m$-ary tree} is a rooted tree in which each internal node has exactly $m$ children. The terms \textbf{bifurcating} and full binary are also used when $m = 2$, in which case every internal node is associated with a left and a right subtree, of which its two children are the roots. If $m=1$ then the tree is \textbf{linear}.
\end{definition}

\begin{definition}[Cherry]
	A tree consisting of only a root and two leaves is a \textbf{cherry}.
\end{definition}

\begin{definition}[Caterpillar tree]
	A \textbf{caterpillar tree} is a bifurcating tree in which every internal node except one has exactly one child leaf.
\end{definition}

% \begin{definition}[Fully symmetric or complete tree]
% 	If, for every internal node $i$, the subtrees rooted at the children of $i$ all contain the same number of leaves then the tree is \textbf{fully symmetric}. Computer scientists sometimes call this a \textbf{complete} tree.
% \end{definition}
 
\subsection{The universal tree balance index $J^1$}

In \citet{lemant_robust_2022} we introduced a general class of tree shape indices of the form
\begin{equation}\label{Jdef}
	J(T) = \frac{1}{\sum_{k \in \tilde{V}(T)} g_k} \sum_{i\in\tilde{V}(T)} g_i W_i,
\end{equation}
where $g_k$ is a node importance factor and $W_i$ is a node balance score. We showed that if the importance factors and balance scores satisfy certain conditions then $J$ can, in a precise sense, be considered a robust, universal tree balance index.
Within this general class, we further defined a particular index $J^1$ based on a node balance score function $W^1$.
For every child $j$ of an internal node $i$, the node balance score $W_{ij}^1$ is defined in terms of the normalized Shannon entropy:
\begin{equation}\label{Wij1}
	W_{ij}^1 = 
	\begin{cases}
		-\frac{S_j}{S_i^*}\log_{d^+(i)}\frac{S_j}{S_i^*}, & \text{ for } d^+(i) > 1 \\
		0,                                                & \text{ otherwise,}      
	\end{cases}
\end{equation}
where $S_i$ is the magnitude of the subtree rooted at node $i$, including $i$, and $d^+(i)$ is the outdegree (number of children) of $i$. The tree balance index $J^1$ is then defined as a weighted mean of node balance scores by setting $g_i = S_i^*$ and $W_i = \sum_{j\in C(i)} W_{ij}^1$:
\begin{equation}\label{J1def}
	J^1(T) = \frac{1}{I_{S,\text{gen}}(T)} \sum_{i\in\tilde{V}(T)} S_i^{*}\sum_{j\in C(i)}W_{ij}^{1},
\end{equation}
where $C(i)$ is the set of children of node $i$. Note that, for all $i \in C(i)$, we have $0 \le W_i \le 1$ and so $J^1$, being a weighted average of $W_i$ values, has extrema zero (minimally balanced) and one (maximally balanced).
		
In \citet{lemant_robust_2022} we proved that, in the case of full $m$-ary trees, $J^1$ is identical to the normalized reciprocal of the generalized Sackin index:
\begin{proposition}[The leafy tree identity \citep{lemant_robust_2022}] \label{proposition6}
	Let $T$ be a leafy tree with $d^+(i) = m > 1$ for all internal nodes $i$. Then
	\begin{equation}
		J^1(T) = \frac{H_m(T)S(T)}{I_{S,\text{gen}}(T)}, \label{prop6}
	\end{equation}
	where $H_m(T) = -\sum_{i \in L(T)} \frac{f(i)}{S(T)} \log_m \frac{f(i)}{S(T)}$ is the Shannon entropy (base $m$) of the leaf sizes $f(i)$ divided by their sum $S(T)$. If additionally the leaves are equally sized then $J^1(T) = n \log_m n / I_S(T)$.
\end{proposition}
Moreover, we proved that this property is unique to $J^1$ among all indices defined as weighted means of node balance scores. Following \citet{noble_new_2023}, we will refer to Proposition~\ref{proposition6} as the leafy tree identity.
 
\subsection{$J^1$ unites and generalizes prior notions of tree balance}

In computer science, tree balance is effectively a binary property: a tree is considered balanced if its height is sufficiently small, given its leaf count. A precise definition for binary trees is that the height of the left subtree of every node must be no more than one branch longer or shorter than the height of its right subtree \citep{knuth1997art_vol3} (p. 459). In biology, where comparisons between trees are more relevant, it is conventional instead to use a normalized form of the Sackin index or another index to assign balance values on a continuum \citep{Shao1990, fischer2023tree}. Subsequent to defining $J^1$ in \citet{lemant_robust_2022}, we have discovered that our index uniquely connects these two historically separate notions of tree balance.

Key to establishing this connection is the concept of a self-balancing binary search tree. Binary search trees are data structures in which each node is associated with an item of data and information is retrieved by following a path down the tree from the root to the relevant node. The expected time taken to retrieve (or add or remove) data from a binary search tree is thus proportional to the internal path length, and the worst-case time taken is proportional to the tree height. A self-balancing binary search tree is one that automatically adjusts itself to maintain a relatively small height as nodes are added and deleted. \citet{nievergelt_binary_1972} introduced one of the most popular types of self-balancing binary search tree, originally called a tree of bounded balance but now more commonly known as a weight-balanced tree \citep{knuth1997art_vol3} (p. 476). Their crucial insight was that the height and the internal and external path lengths can be kept small by ensuring that, for all internal nodes, the ratio of the left subtree leaf count and the right subtree leaf count stays sufficiently close to unity.
    
The principles underlying the weight-balanced tree concept were presented in \citet{wong_upper_1973}, which derived upper bounds for the internal and external path lengths of binary trees. In proving these results, \citet{wong_upper_1973} defined what they called the entropy of a node. They consequently defined the average entropy of a tree and stated an identity relating this average entropy to the external path length.
The entropy of a node is identical to our node balance score in the special case of full $m$-ary leafy trees with uniform leaf sizes. Average entropy is (ignoring an erroneous division sign in their definition) equivalent to $J^1$ restricted to the same special case. Moreover, the identity of \citet{wong_upper_1973} is equivalent to the special case of the leafy tree identity (Proposition~\ref{proposition6}). A seldom cited paper by \citet{nievergelt_bounds_1972} implicitly extended these concepts to account for arbitrary node sizes. The universal tree balance index that we call $J^1$ was thus defined -- and then effectively forgotten -- fifty years before we independently rediscovered it.
	
\subsection{$J^1$ is maximized by Huffman coding}

Huffman coding provides another way to understand the equivalence between minimizing binary tree path length and maximizing $J^1$.
\begin{definition}[Huffman coding, \citet{huffman_method_1952}]
	\textbf{Huffman coding} is a method for constructing a leafy binary tree with minimal weighted path length, given a set of leaves with arbitrary sizes $\alpha_1, \dots, \alpha_n > 0$. We begin with the set $B = \{ t_1, \dots, t_n \}$, where each $t_i$ is a tree comprising exactly one node of size $\alpha_i$. Without loss of generality, let $t_1$ and $t_2$ be the two smallest magnitude trees in $B$. We then combine $t_1$ and $t_2$ by attaching them to a root of size zero to obtain $B = \{ t_{1,2}, \dots, t_n \}$. We repeat this process until $B$ comprises a single tree $t_{1,2,\dots,n}$, called the \textbf{Huffman tree}.
\end{definition}
\begin{proposition}
	The Huffman tree maximizes $J^1$ on bifurcating leafy trees for a given set of leaf sizes.
\end{proposition}
\begin{proof}
	By the leafy tree identity, the Huffman method maximizes $J^1$ as it minimizes the weighted path length.
\end{proof}
	
\subsection{Expected values of $J^1$ under standard null models}\label{expsection}
For applications in evolutionary biology, it is important to know the expected values of a balance index under simple null models. These expected values can then be compared with values obtained for trees inferred from empirical data, to test hypotheses about the nature of evolution \citep{heard1992patterns, mooers_inferring_1997, mckenzie_distributions_2000, aldous_stochastic_2001, steel_properties_2001}. The simplest and most widely studied null models are the uniform model and the Yule model (or pure birth process), both of which generate bifurcating unlabelled trees. 

\begin{figure}[h]
	\centering
	\includegraphics[width=0.8\textwidth]{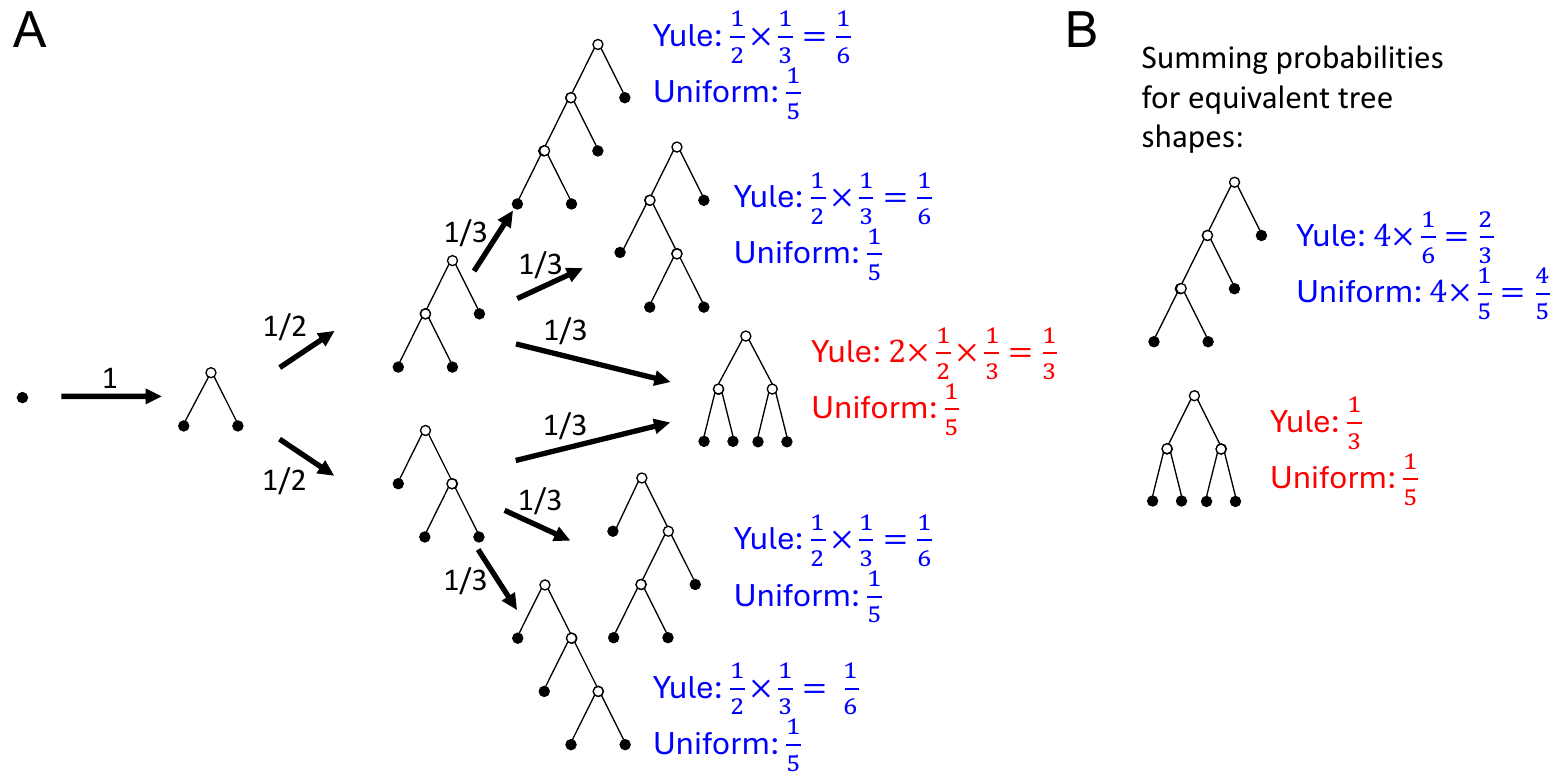}
	\caption{\textbf{A:} Tree generation under the Yule model. Each four-leaf tree is labelled with the probabilities it is assigned by the Yule and uniform models. \textbf{B:} Tree shape probabilities are obtained by summing over trees that can be transformed into one another by rotating branches, such as the four asymmetric trees in panel A (blue labels).}
	\label{yule-unif-figure}
\end{figure}
    
\begin{definition}[Yule model, \citet{yule_iimathematical_1925}]\label{yuledef}
	The \textbf{Yule model} is associated with the following process: Starting with a cherry, repeatedly replace one leaf, chosen uniformly at random, with a cherry, so that the number of leaves increases by one at each step (Figure~\ref{yule-unif-figure}A). For a given number of leaves, each tree is assigned a probability proportional to the number of ways it can be generated by this process.
\end{definition}
\begin{definition}[Uniform model, \citet{rosen_vicariant_1978}]\label{unifdef}
	Under the \textbf{uniform model}, every bifurcating tree on $n$ leaves is assigned the probability $n\binom{2n-2}{n-1}^{-1}$, the inverse of the number of distinct trees on $n$ leaves.
\end{definition}
Tree balance indices assign equal values to trees that can be transformed into one another by rotating branches. We can therefore sum probabilities over all trees with equivalent shapes (Figure~\ref{yule-unif-figure}B). Since the Yule and uniform models do not assign node sizes, we will consider only leafy trees with uniform leaf sizes. In systematic biology, such trees can be interpreted as cladograms in which leaves represent extant types and internal nodes represent extinct common ancestors \citep{podani2013tree}.

The leafy tree identity implies that the expected value of $J^1$ for full $m$-ary leafy trees with $n$ equally sized leaves is
\begin{equation}
	\mathbb{E}(J^1) = \mathbb{E} \left( \frac{n\log_2 n}{I_S} \right) 
	%= n\log_2 n  \,\mathbb{E}_Y^n \left( \frac{1}{I_S} \right) 
	= \frac{n\log_2 n} {\mathbb{H} \left( I_S \right)},
\end{equation}
where $\mathbb{H}(I_S) = 1 / \mathbb{E}(1 / I_S)$ is the harmonic mean of the Sackin index. Obtaining the expected value of $J^1$ for such trees under any model is therefore equivalent to obtaining the harmonic mean of the Sackin index under the same model. 
Although we have yet to find a closed-form expression for the harmonic mean of $I_S$ under either of the standard null models, we have obtained exact expected values of $J^1$ for $n\leq 11$ (Yule model) and $n\leq 10$ (uniform model) by iteratively generating all possible $n$-leaf trees (Tables~\ref{table_values} and~\ref{table_values_unif}; black crosses in Figure~\ref{jensenfig}).

\begin{figure}[h]
	\centering
	\includegraphics[width=\textwidth]{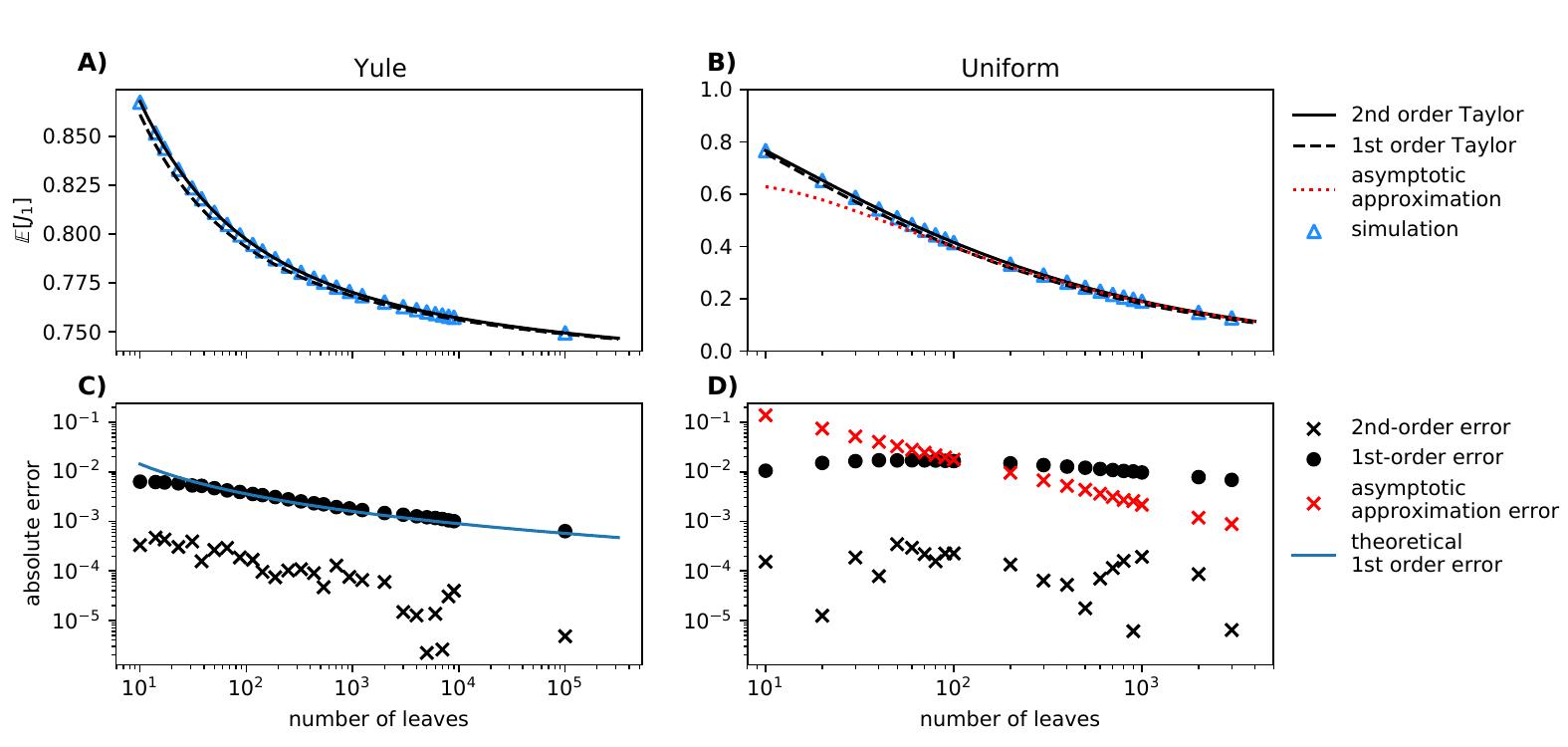}
%   \iffalse
% 	\caption{\textbf{Top row}: Exact and approximate expected values of $J^1$ for the uniform model (\textbf{A}) and the Yule model (\textbf{B}). 
%     Black crosses show exact values of $\mathbb E(J^1)$ (Tablew~\ref{table_values} and~\ref{table_values_unif}). 
%      Red crosses are $J^1$ sample means, with each sample comprising $100,000$ randomly generated trees; error bars show standard errors.
%      Grey points show the approximation $n\log_2n/\mathbb E(I_S)$.
%      Orange and blue curves show upper and lower bounds on $\mathbb E(J^1)$ derived in Appendix~\ref{proof-app}.
%     \textbf{Bottom row}: Errors in the approximation of $\mathbb E(J^1)$ under the uniform model (\textbf{C}) and the Yule model (\textbf{D}). 
%     Black crosses show the exact Jensen gap $\mathcal{J}(n)$.
%     Red crosses are sample mean Jensen gaps, with each sample comprising $100,000$ randomly generated trees; error bars show standard errors.
%     Orange and blue curves show upper and lower bounds on $\mathcal{J}(n)$, as derived in Appendix~\ref{proof-app}.
%     {In the legends, I need to change ``min" and ``max" to ``lower bound" and ``upper bound".}
%     \arm{On the same vain as the comments I give in the other manuscript, should we avoid e-03 notation} {Yes, please use nicer notation.}}
% \fi
    \caption{$\mathbb{E}[J^1]$ approximations along with associated errors. \textbf{Top row}: Blue triangles show the sample mean of $J^1$ under the Yule \textbf{(A)} and uniform \textbf{(B)} tree generation processes . 1st and 2nd order Taylor approximations are shown for both processes. The 1st order approximation also corresponds to a harmonic mean-expectation approximation. In the uniform model, an approximation based on the asymptotic distribution of $X_n=\frac{I_S}{n^{3/2}}$ is shown in red. At each of the designated leaf counts, we sample $10^6$ trees under the Yule or uniform process. We use the straightforward recursion equation to sample Yule trees (Eq. \ref{eq:rec}), and we employ the random bracket sequence method described in \citet{atkinson1992binary} to sample uniform binary trees. \textbf{Bottom row}: The errors of each approximation type (Yule \textbf{(C)} and uniform \textbf{(D)}) are shown with respect to the tree size. The blue line indicates the theoretically expected error. For further details of the approximations see Appendix~\ref{app:approx_j1}.}
	\label{jensenfig} 
\end{figure}
	
We can also approximate the expected value of $J^1$ by using the arithmetic mean to approximate the harmonic mean: $\mathbb{H}(I_S) \approx \mathbb{E}(I_S)$.
Under the Yule and uniform models, the expected values of the Sackin index for trees on $n$ leaves are respectively
\begin{equation} \label{yule_exp_sackin}
	\mathbb{E}_Y(I_S) = 2n\sum_{i=2}^n \frac{1}{i} = 2n (H_n - 1), \quad \mathbb{E}_U(I_S) = n\left( \frac{(2n-2)!!}{(2n-3)!!}-1 \right),
\end{equation}
where $H_n$ is the $n$th harmonic number and $k!!$ is the double factorial $k!! = k(k-2)(k-4)\dots$ \citep{kirkpatrick_searching_1993, mir_new_2013}.
Hence 
\begin{equation} \label{approxEJ1}
	\mathbb{E}_Y(J^1) \approx \frac{\log_2 n}{2(H_n - 1)}, \quad \mathbb{E}_U(J^1) \approx \frac{(2n-3)!! \log_2 n}{(2n-2)!!-(2n-3)!!}.
\end{equation}
The error in each approximation is the Jensen gap
\begin{equation} 
\mathcal{J}(n) = \mathbb E(J^1)-\frac{n \log_2 n}{\mathbb E(I_S)}.
\end{equation} 
In \citet{lemant_robust_2022} we obtained numerical results, based on samples of randomly generated trees, suggesting that $\mathcal{J}(n)$ is relatively small for $n \le 128$.

A result of \citet{liao_sharpening_2017} leads to the following more rigorous and precise new results:
\begin{proposition}\label{jensen_prop}
	%Let $\mathbb{E}_Y(J^1)$ and $\mathbb{E}_U(J^1)$ be expectation values of $J^1$ under the Yule and uniform models respectively. Then
    For bifurcating leafy trees with $n$ equally sized leaves:
	\begin{enumerate}
		\item For the Yule model, $\mathcal{J}(n) < 0.008$ for all $n$, and $\mathcal{J}(n) \to 0$ as $n\to\infty${, which implies that, asymptotically, $J^1\approx \frac{1}{2\ln 2} \approx 0.72$ (in the sense that the distribution converges to this value as $n\to\infty$).}
		\item For the uniform model, $\mathcal{J}(n) < 4/(3 \pi e^2) \approx 0.057$ for all $n$.
	\end{enumerate}
\end{proposition}
\begin{proof}
     See Appendix~\ref{jensen_proof}.
\end{proof}

{The first part of Proposition~\ref{jensen_prop} can be generalized to the Yule process for full $m$-ary trees (see Appendix~\ref{Yule_m-ary}).}

The errors in our first-order approximation (black dots in Figure~\ref{jensenfig}C and~\ref{jensenfig}D) lead us to conjecture further that $\mathcal{J}(n)$ for the uniform model is never more than 0.02 and approaches zero as $n\to\infty$. The gap is larger for the uniform model than for the Yule model because $I_S$ has greater variance (Figure~\ref{var_fig}).

The approximation $\mathbb{E}[\frac{1}{I_S}]\approx \frac{1}{\mathbb{E}[I_S]}$ can also be interpreted as the moment of a first-order Taylor approximation of the function $f(x)=\frac{1}{x}$ composed with the random Sackin index. Using this approach, we obtain more precise approximate asymptotic error expressions in the Yule case and find that $\mathcal{J}(n)$ decays at an asymptotic approximate rate of $(\ln n)^{-2}$. We can further perform a second-order Taylor expansion:
\begin{equation}\label{eq:2nd_order}
    \mathbb{E}[J^1]\approx \frac{n\log_2 n}{\mathbb{E}[I_S]}+\frac{n\log_2 n}{\mathbb{E}[I_S]^3}\mathbb{V}[I_S],
\end{equation}
where $\mathbb{V}$ denotes the variance. This approximation decays at a faster rate $(\ln n)^{-3}$. In the uniform model in Figure~\ref{jensenfig}, we empirically observe that equation \ref{eq:2nd_order} is an accurate approximation of $\mathbb{E}[J^1]$. An asymptotic approximation based on the asymptotic distribution of $X_n=\frac{I_S}{n^{3/2}}$ in the uniform model is conjectured to decay at polynomial rate (red crosses in Figure \ref{jensenfig}d). All these results are derived in Appendix~\ref{app:approx_j1}.

\subsection{Least balanced leafy trees with equally sized leaves}

We now turn from expected to minimal values. We previously proved that $J^1$ attains its minimal value 0 if and only if the tree is linear, and its maximal value 1 if and only if each internal node splits its descendants into at least two subtrees of equal magnitude \citep{lemant_robust_2022}. Because $J^1$ accounts for node sizes, any non-linear tree, regardless of topology, can be made arbitrarily unbalanced ($J^1 \to 0$) or arbitrarily balanced ($J^1 = 1$) by adjusting its relative node sizes (Figure~\ref{balcat}A). It is therefore trivial to find the most and least balanced trees across all possible node sizes. The most important open questions concern the extremal values of $J^1$ when the node degree and node size distributions are somehow constrained.

The simplest interesting case is that of leafy trees with equally sized leaves. Among bifurcating trees with a given number of leaves, the caterpillar tree maximizes Sackin's index \citep{fischer_extremal_2021}. Consequently, the leafy tree identity implies that, among all bifurcating leafy trees with a given number of equally sized leaves, the caterpillar tree minimizes $J^1$. However, we previously presented a counterexample showing that the caterpillar tree does not always minimize $J^1$ when larger outdegrees are permitted  \citep{lemant_robust_2022}. Specifically, we showed that, among all leafy trees with six equally sized leaves and no nodes of outdegree 1, the least balanced tree according to $J^1$ is not a caterpillar but belongs to a superset of caterpillar trees that we shall call broom trees.

\begin{figure}[h]
	\centering
	\includegraphics[width=0.9\textwidth]{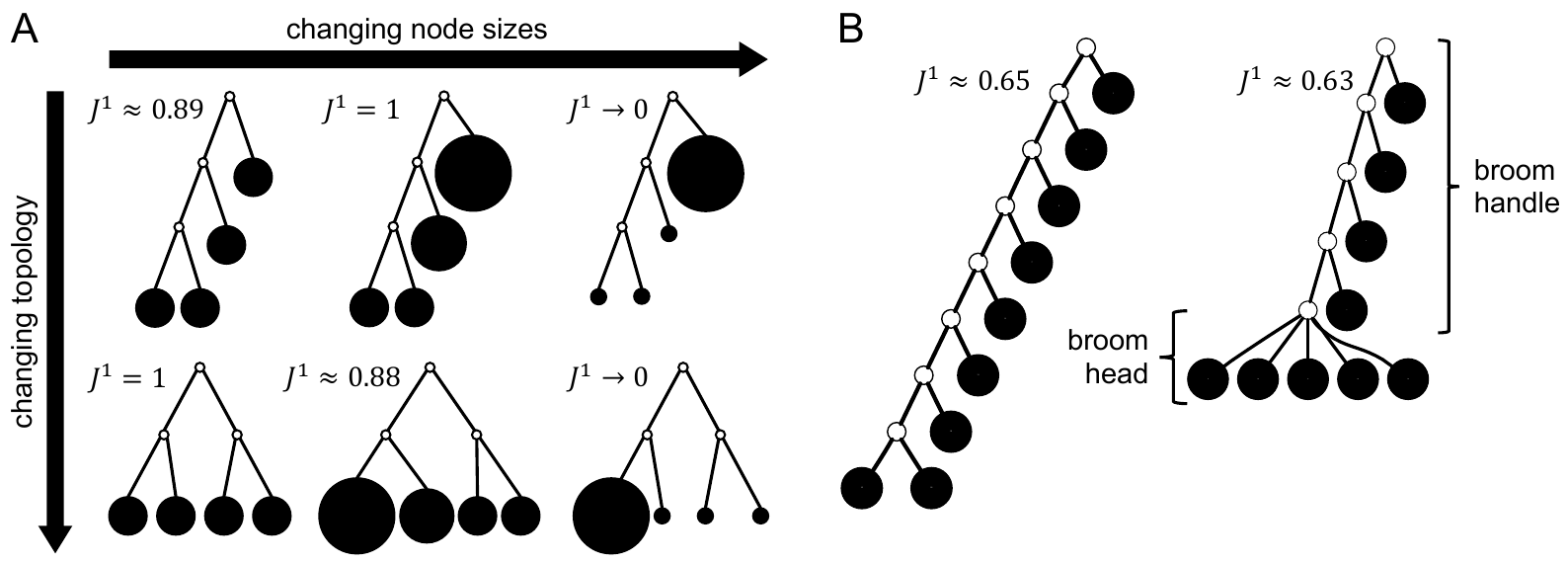}
	\caption{\textbf{A:} By varying relative node sizes, any non-linear tree can be made maximally or minimally balanced according to $J^1$, as illustrated here with a leafy four-leaf caterpillar tree (top row) and a leafy four-leaf tree with symmetric topology (bottom row). The leaf sizes in the first column are all equal and in the second column are $4, 2, 1, 1$. \textbf{B: } The leafy caterpillar and minimally balanced leafy broom tree on nine equally sized leaves. In all trees, open circles represent internal nodes of size zero.}
	\label{balcat}
\end{figure}
	
\begin{definition}[Broom tree]
	A \textbf{broom tree} is such that every internal node except the most distant from the root has outdegree $2$. The \textbf{broom head} is the star subtree rooted at the lowest internal node; the remainder of the tree is the \textbf{handle} (Figure~\ref{balcat}B).
\end{definition}

It is straightforward (see Appendix~\ref{J1Tb_proof}) to derive a general expression for $J^1$ for leafy broom trees with $n$ equally sized leaves, of which $k$ are in the head:
	\begin{equation}
		J^1_B(n, k, 1) := \frac{2\left( n \log_2 n - k \log_2 k + k \right)}{(n+k)(n-k+1)}.\label{J1Tb}
	\end{equation}
By varying the value of $k$ in this expression, we can verify, for example, that the caterpillar is not the minimally balanced leafy broom tree with nine equally sized leaves (Figure~\ref{balcat}B). The following proposition generalizes this result. 
	
\begin{proposition}\label{cat_prop}
	Among leafy trees with $n$ equally sized leaves and no nodes of outdegree 1, the caterpillar minimizes $J^1$ if and only if $n \le 4$.
\end{proposition}
\begin{proof}
        If $n = 2$ then the proposition holds as the caterpillar is the only tree, and for $n = 3$ it holds because the caterpillar is less balanced than the star tree, which is the only alternative.
	From Equation~\ref{J1Tb} for $n > 3$ we find that
	$J^1_B(n,2,1) < J^1_B(n,3,1)$ if and only if
	\begin{equation}
		4n\log_2n + 3(1 - \log_2 3)(n+2)(n-1) > 0.
	\end{equation}
	Given $n > 3$, this inequality is satisfied if and only if $n$ is less than approximately 4.17. As $n$ must be a positive integer, the only valid solution is $n = 4$.
\end{proof}

{
The non-optimality of caterpillar trees can be understood by considering how the node balance scores and weights change when we remove the second lowest internal node of a caterpillar tree and reattach its child leaf to the lowest internal node.
First, the normalizing factor of $J^1$ (which is identical to Sackin's index) decreases.
Second, while the internal nodes that remain in the broom handle have the same balance scores and weights as before the change, the weight assigned to the root of the broom head, which remains maximally balanced, increases due to its additional leaf.
Although assigning more weight to the broom head and reducing the normalizing factor both increase $J^1$, these effects are counteracted by the removal of the handle's most balanced node.
The latter effect dominates when the leaf count is greater than 4 (see Appendix~\ref{app:catexp} for a more detailed explanation).
}

{
The analysis of broom trees illustrates the general principle that when two trees differ in outdegree distribution, $J^1$ accounts for differences in node balance scores and weights as well as in node depths. By contrast, if two leafy trees are both full $m$-ary and have the same number of leaves and the same leaf size distribution, then the difference in $J^1$ is solely determined by the difference in Sackin's index (because, by Proposition~\ref{proposition6}, the weighted sum of the node balance scores must be the same for both trees).
}

We further make the following conjecture, which we have exhaustively verified for trees with 12 or fewer leaves. 
\begin{conjecture} \label{conj_brooms}
 Among leafy trees with a given number of equally sized leaves and no nodes of outdegree 1, the tree that minimizes $J^1$ is a broom tree.
\end{conjecture}
	
\subsection{Least balanced leafy broom trees in the large $n$ limit}

The obvious question arising from Proposition~\ref{cat_prop} and Conjecture~\ref{conj_brooms} is this: Which leafy broom tree is the least balanced for a given leaf count? To examine whether the answer is sensitive to our assumption of equal leaf sizes, we will slightly relax this assumption by considering leafy broom trees in which all leaves in the head have the same size, and all leaves in the handle have the same size, while allowing these two sizes to differ. The task is then to investigate $k^* = \arg \min_{2 \le k \le n} J^1_B(n,k,p)$, where $n$ is the total number of leaves, $k$ is the number of leaves in the head, $p > 0$ is the size of each leaf in the head, relative to the size of each leaf in the handle, and
    \begin{equation}\label{J1B} 
        J^1_B(n, k, p) := \frac{2[kp+(kp+n-k)\log_2(kp+n-k) - kp \log_2(kp)]}{(2kp+n-k)(n-k+1)}.
    \end{equation}
Note that, although $k^*$ depends on $n$ and $p$, we will tend to avoid writing $k^*_{n,p}$ to simplify our notation.

\begin{figure}[h]
	\begin{center}
		\includegraphics[width = \textwidth]{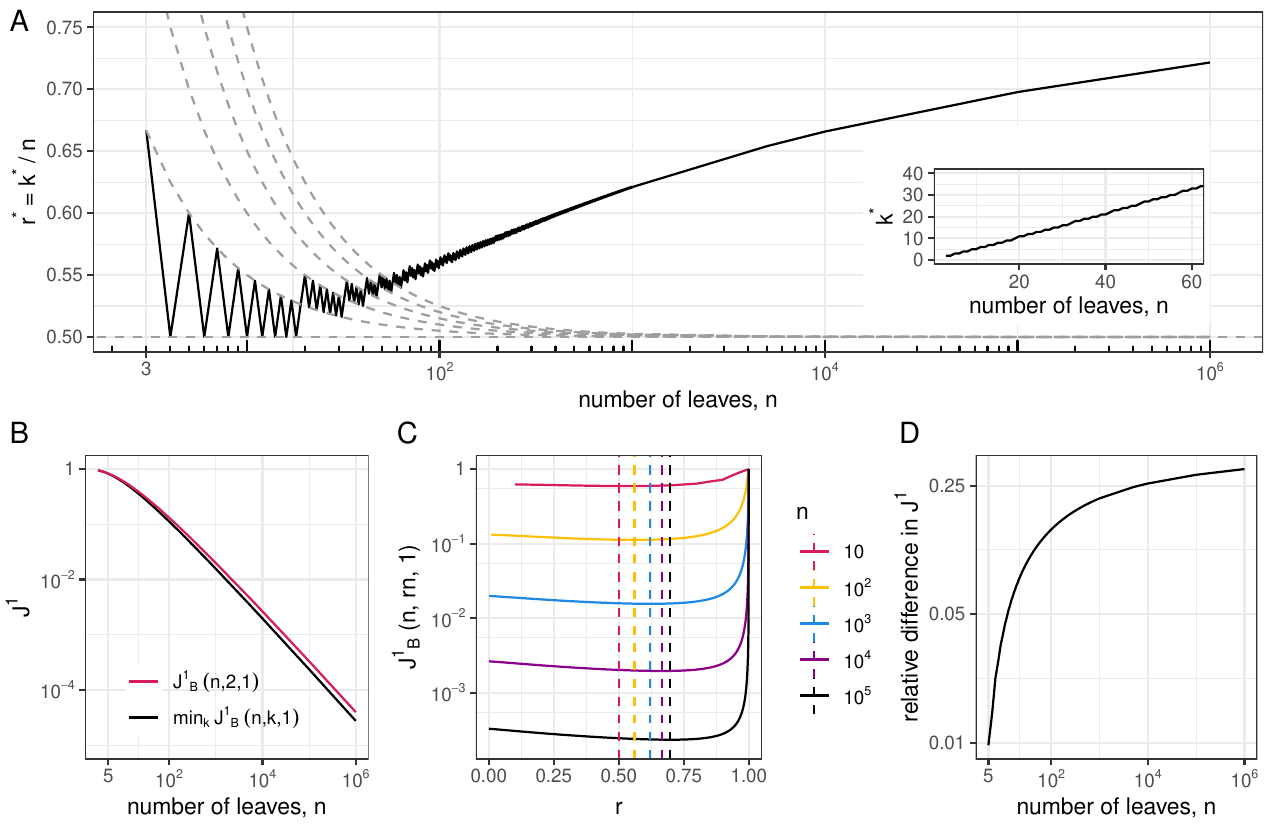}
		\caption{ Numerical results for least balanced leafy broom trees with equally sized leaves. \textbf{A:} Values of $r = k/n$ (main plot) and $k$ (inset) that minimize $J^1_B(n, k, 1)$. %\iffalse\yan{The small, inserted figure, does not seem consistent with the large one for small n.} {I don't see the problem.} \arm{In (A), the graph zigs up and down, and then incremenets, and then zigs up an down, etc. In the inset, there is a monotonic increase. Although, this is explanable: the y-axis in the inset is $r^*\cdot n(=k^*)$ rather than $r^*$} {Here are the first 25 y-values in the inset: 2, 2, 3, 3, 4, 4, 5, 5, 6, 6, 7, 7, 8,  8,  9,  9, \textbf{10}, 11, 11, 12, 12, 13, 13, 14, 14. You can check these are equal to the $r^*$ values in the main plot multiplied by the corresponding $n$ values. The single 10 (in bold) corresponds to a jump up in the main plot and a barely noticeable kink in the inset (look at inset\_plot.pdf in the figures folder).}\fi 
        The discrete values are connected by solid black lines to guide the eye. Dashed grey curves in the main plot show $r^* = \frac{n+a}{2n}$ for $a \in \{ 0,1,2,3,4,5 \}$. These curves intersect with the $r^*$ values of trees for which $k^* = \frac{n+a}{2}$. \textbf{B:} Exact $J^1_B$ values for caterpillar trees and minimally balanced leafy broom trees with equally sized leaves, as a function of $n$. \textbf{C:} $J^1_B(n,rn,1)$ versus $r$ for different values of $n$. The dashed lines are at the values of $r$ that minimize $J^1$. \textbf{D:} The relative difference in $J^1$ values between the minimally balanced broom tree and the caterpillar tree, as a function of $n$.}
		\label{Rfigures}
	\end{center}
\end{figure}

Among leafy broom trees with equally sized leaves (that is, $p = 1$), a numerical investigation of Equation~\ref{J1B} reveals that that $r^* = k^*/n$ is never less than $\frac{1}{2}$ and oscillates while generally increasing with $n$ (Figure~\ref{Rfigures}A).
We can understand the upward trend in $r^*$ by considering the large tree limit.

% \begin{proposition} \label{largebrooms}
% The following hold as $n \to \infty$. \yan{This is not a precise mathematical statement!} {I propose Yannick rewrites this proposition.} \arm{Why not just have this: in the $p\leq \frac{1}{2}$ case, for the minimizer, pre-empt with there exists some $N$ (or being pedantic, $N_p$ [just to show that this large enough criteria depends on $p$]) such that for all $n>N$, $k^*=2$ [and other stuff]. For the other expressions, write, ``as functions of $n$'', [stuff]}

% If $0 < p \le \frac{1}{2}$ then
% \begin{gather}
%     k^* = 2, \text{ hence} \min_{2 \le k \le n} J^1_B(n,k,p) = J^1_B(n, 2, p), \\
%     J^1_B(n,2,p) \sim \frac{2 \log_2 n}{n}.
% \end{gather}

% If $p > \frac{1}{2}$ then
% \begin{gather}
%     1 - r^* \sim \frac{\sqrt{2} p} {\sqrt{(2p-1)\log_2{n}}},\label{arg_min_growth_large_p} \\
%     \min_{2 \le k \le n} J^1_B(n,k,p) \sim \frac{\log_2 n}{np}.
% \end{gather}
% \end{proposition}

% \begin{proof}
%     See Appendix.
% \end{proof}

%\arm{Something like this?}
\begin{proposition} \label{largebrooms}
Treating $J_B^1$ as a function of $n$, the following hold:

If $0 < p \le \frac{1}{2}$ then there exists $N_p$ such that for all $n>N_p$, $k^* = 2$. Hence
\begin{gather}
    \min_{2 \le k \le n} J^1_B(n,k,p) = J^1_B(n, 2, p) \sim \frac{2 \log_2 n}{n}.
\end{gather}

If $p > \frac{1}{2}$ then
\begin{gather}
    1 - r_n^* \sim \frac{p\sqrt{2}} {\sqrt{(2p-1)\log_2{n}}},\label{arg_min_growth_large_p} \\
    \min_{2 \le k \le n} J^1_B(n,k,p) \sim \frac{\log_2 n}{np}.
\end{gather}
\end{proposition}

\begin{proof}
    See Appendix~\ref{largebrooms_proof}.
\end{proof}

Proposition~\ref{largebrooms} implies that the caterpillar is the least balanced broom tree when $p \le \frac{1}{2}$ and $n$ is sufficiently large. In contrast if $p > \frac{1}{2}$ then, for sufficiently large $n$, the least balanced broom tree has most of its leaves on the head, with relatively few on the handle. In the case of equally sized leaves ($p = 1$), the absolute difference in $J^1$ values between the minimally balanced broom and the caterpillar is always small and decreases rapidly as the number of leaves increases (Figure~\ref{Rfigures}B,C), yet Proposition~\ref{largebrooms} implies that the relative difference approaches $\frac{1}{2}$ (Figure~\ref{Rfigures}D).

\begin{figure}[h]
	\centering
	\includegraphics[width=\textwidth]{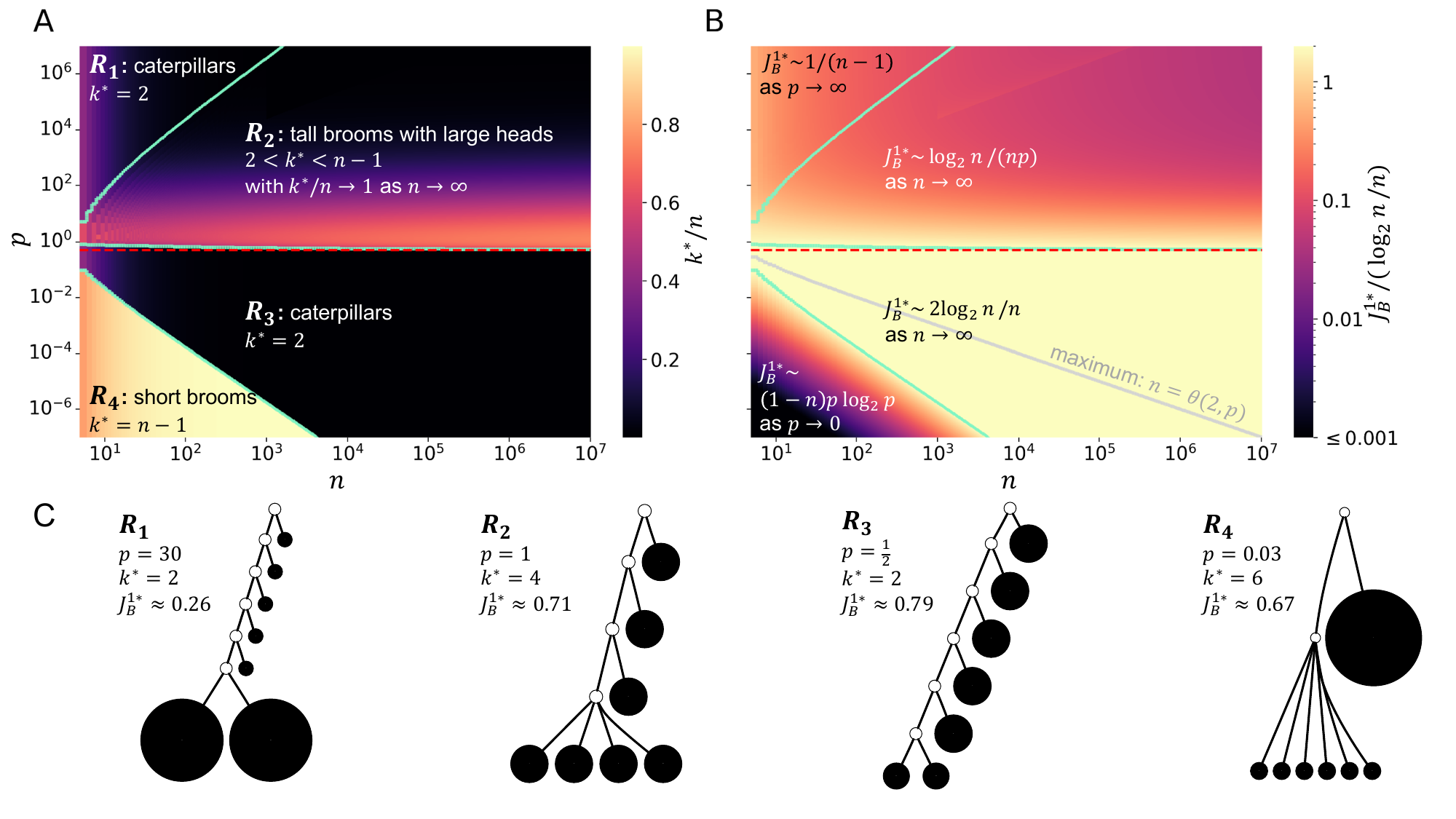}
	\caption{Least balanced broom trees in $n$-$p$ space. \textbf{A:} For $n > 4$, there are four cases of $k^*=\arg\min_{2\leq k\leq n} J_B^1(n,k,p)$, separated here by turquoise curves. The case $k^* = 2$ corresponds to caterpillar trees. The middle curve approaches $p = \frac{1}{2}$ (red dashed line) as $n \to \infty$. \textbf{B:} $J_B^{1*} = \min_{2 \le k \le n-1} J_B^1 (n, k, p)$ relative to $\log_2 n / n$. Each region is labelled with the corresponding asymptotic behaviour of $J_B^{1*}$. The grey curve at $n = \theta(2, p)$ is where $J_B^{1*}$ is maximal (the formula for $\theta(k,p)$ is given in Proposition \ref{broom_J1_in_p}). \textbf{C:} A representative least balanced broom tree for each of the four cases when $n = 7$. Open circles represent internal nodes of size zero.}
	\label{boundaries_plot}
\end{figure}

\subsection{Least balanced leafy broom trees in the general case}

For trees with fewer leaves the picture is more complicated. If $n=3$ then $k^* = 2$ for all $p>0$. For $n=4$ there are two cases: $k^* = 2$ and $k^* = n - 1 = 3$, depending on $p$. For $n > 4$ there are four cases corresponding to regions of $n$-$p$ space that we will label $R_1$ to $R_4$ (Figure~\ref{boundaries_plot}A, \ref{boundaries_plot}C). 
In the non-adjacent $R_1$ and $R_3$, we have $k^* = 2$, which means that the least balanced trees in these regions are caterpillars. In $R_2$, we find $2 < k^* < n-1$ with $k^*/n \to 1$ as $n \to \infty$, so that the least balanced trees have both long handles and large heads. In $R_4$, $k^* = n - 1$, corresponding to trees of height 2.

To locate the regional boundaries, we first note that Proposition~\ref{largebrooms} implies that the curve separating regions $R_2$ and $R_3$ approaches $p = \frac{1}{2}$ as $n \to \infty$. Proposition~\ref{cat_prop} also tells us that $p > 1$ for all $(n,p) \in R_1$ and  $p < 1$ for all $(n,p) \in R_3$. The following proposition further establishes that $p > \frac{1}{2}$ for all $(n,p) \in R_2$ and $p < \frac{1}{2}$ for all $(n,p) \in R_4$. Hence, in summary, the $R_2$-$R_3$ boundary is bounded below by $p=\frac{1}{2}$, bounded above by $p=1$, and approaches the lower bound as $n \to \infty$.
%Equivalently, this proposition tells us that, for any given $n>3$, the caterpillar is the least balanced broom tree with $p = \frac{1}{2}$.
%\yan{Why is this equivalent? I am not really following, but intuitively this "equvalently" seemed suspicious.} {I don't see the problem. $k^* = 2$ is equivalent to the tree being a caterpillar tree.}

%\yan{is this for n large enough or for all n? Say k*=2 in regions R1 and R3.} {I don't know which statement you're referring to.}

\begin{proposition}\label{cats_p_one_half}
If $p = \frac{1}{2}$ then $k^* = 2$ for all $n \ge 3$.
\end{proposition}
\begin{proof} See Appendix~\ref{cats_p_one_half_proof}.
\end{proof}

{
Just as we did for the large-$n$ limit, we can also describe the nature of $J^1$ for the least-balanced broom tree in the large- and small-$p$ limits.
}

{
\begin{proposition}\label{asymptotic}
    For any fixed $n \ge 3$:
	\begin{enumerate}
	\item As $p\to \infty$, $J_B^1(n,k^*,p) \sim \frac{1}{n-1}$.
        \item As $p\to 0$, $J_B^1(n,k^*,p) \sim p(1-n)\log_2 p$.
        \end{enumerate}
\end{proposition}
}

\begin{proof}
    See Appendix~\ref{asym_JB_p}.
\end{proof}

Lastly, we show that $R_1$ and $R_4$ contain trees with arbitrarily large leaf counts: 

\begin{proposition}\label{R3_R4_boundary}
    The boundaries between the $R_3$ and $R_4$ regions, and between the $R_1$ and $R_2$ regions, are asymptotically $p=2 / n^2$ and $p=n^2\log_2 n / 3$, respectively.
\end{proposition}

\begin{proof}
    See Appendix~\ref{R3_R4_boundary_proof}.
\end{proof}

{
The following results shed further light on how, for a given number of leaves, the $J^1$ value of the least balanced broom tree -- that is, $\min_{2 \le k \le n-1} J_B^1 (n, k, p) = J_B^1 (n, k^*, p)$ (abbreviated to $J_B^{1*}$ in Figure~\ref{boundaries_plot}B) -- varies with $p$.
}

\begin{proposition}\label{broom_J1_in_p}
For all $n \ge 3$ and $k \in \{2,\dots, n-1 \}$, $J^1_B(n,k,p)$ as a function of $p > 0$ is strictly decreasing if $n > \theta(k, p)$ and strictly increasing if $n < \theta(k, p)$, where $\theta(k, p) = \frac{2}{kp} + k(1 - p)$.
\end{proposition}

\begin{proof}
Since $n > k$, 
    \begin{align*}
    \sign \left( \frac{\partial J^1_B}{\partial p} (n,k,p) \right) &= 
    \sign \left( \frac{2k(n-k) (1 - \log_2 [kp(n - k(1 - p))])}{(2kp+n-k)^2(n-k+1)} \right) \\ 
    &= \sign(\theta(k,p) - n).
    \end{align*}  
    % \begin{align*}
    % \sign \left( \frac{\partial J^1_B}{\partial p} (n,k,p) \right) &= \sign (1 - \log_2 [kp(n + k(p - 1))]) \\
    % &= \sign (2 - kp(n + k(p - 1))) \\
    % &= \sign(kp(\theta(k,p) - n)) \\
    % &= \sign(\theta(k,p) - n).
    % \end{align*}
\end{proof}

\begin{proposition}\label{min_J1_in_p}
    For all $n \ge 3$, $\min_{2 \le k \le n} J^1_B(n,k,p)$ as a function of $p > 0$ is strictly decreasing if $n > \theta(2, p)$ and strictly increasing if $n < \theta(2, p)$, where $\theta(2, p) = \frac{1}{p} + 2(1 - p)$.
\end{proposition}
\begin{proof}
    See Appendix~\ref{min_J1_in_p_proof}.
\end{proof}

\begin{corollary} \label{max_min_J1B}
For all $n \ge 3$, 
    $$
    \min_{2 \le k \le n} J^1_B(n,k,p) \le 
    \frac{2 \left(2 - \log_2 \left(\sqrt{n^2-4n+12} - n + 2 \right) \right)}{n - 1},
    $$
    with equality if and only if $n = \theta(2, p)$.
\end{corollary}
\begin{proof}
    See Appendix~\ref{max_min_J1B_proof}.
\end{proof}

The curve $n = \theta(2, p)$, where $\min_{2 \le k \le n} J^1_B(n,k,p)$ is maximal, is shown in Figure~\ref{boundaries_plot}B (solid grey curve). The three panels of Figure~\ref{boundaries_plot} together summarize all the main results of this section.

% \textcolor{red}{I don't think we really need the following corollary any more. Our new results are more interesting.}
% \begin{corollary} \label{catwins}
% If $0 < p \le \frac{1}{2}$ then $\exists N\in\mathbb{N}$ such that for all $n>N$, $J_B^1(n,2,1) < \min_k J_B^1(n,k,p)$, for all $k \in \{2,\dots,n\}$.
% \end{corollary}
% \begin{proof}
%     From Proposition~\ref{largebrooms}, $\arg \min_{2 \le k \le n} J_B^1(n,k,p) = 2$ as $n \to \infty$. 
%     Since $p < 1$, the result then follows from Proposition~\ref{broom_J1_in_p}.
% \end{proof}

% Corollary~\ref{catwins} tells us that, although $J_B^1(n,2,1)$ and $\min_k J_B^1(n,k,p)$ approach the same limit as $n \to \infty$, the caterpillar tree with equally sized leaves is less balanced than any broom tree with $p \le \frac{1}{2}$.
	
\section{Discussion}

The aims of this paper were to explore mathematical properties of the $J^1$ index and to solidify its status as a universal tree balance index. By extending past results and uncovering new connections, we have shown that $J^1$ unifies the notions of tree balance from biology and computer science.
%\arm{What is meant by ``uniquely''?}.
We have set the historical record straight by explaining how the ``average entropy of a tree" definition of \citet{wong_upper_1973} is equivalent to a special case of our $J^1$ definition. Yet whereas, throughout the five decades before we independently rediscovered it, this index was regarded merely as an accessory to defining weight-balanced trees, we instead contend that $J^1$ is useful in its own right as the best index for quantifying tree balance \citep{lemant_robust_2022, noble2022spatial, noble_new_2023}.

While closed-form expressions remain elusive, we have proven the accuracy of simple approximations to the expected values of $J^1$ for the two most important tree generating models. These results put $J^1$ on the same footing as the conventional Sackin and Colless tree balance indices, for which the corresponding expected values have been derived previously \citep{kirkpatrick_searching_1993, blum2006mean, mir_new_2013}. The errors in our approximations are small enough as to be negligible in many practical applications. When this is not the case, statistical methods can employ the maximum and minimum bounds on the true expected values (Figure~\ref{jensenfig}).

We have also investigated minimal values of $J^1$ in the interesting special case of broom trees (in which every internal node except the most distant from the root has outdegree 2). 
%Extremal values of tree balance indices are studied for three main reasons. First, because conventional indices vary with leaf count and outdegree as well as balance, it is necessary to normalize them by subtracting the minimum attainable value and dividing by the range over the set of all trees with the same leaf count and outdegree as the trees of interest \citep{Shao1990, coronado2020minimum, fischer_extremal_2021}. An important advantage of $J^1$ is that it requires no such normalization, and so this issue does not concern us.
%The second motivation is that every tree balance index must, by definition, attain its extremal values only on trees that can reasonably be described as extremely balanced or extremely imbalanced. Lastly, knowing the tree types for which an index attains its maximum and minimum values aids interpretation of the index, by providing reference points to which other trees can be compared.
{
A motivation for this investigation is that every tree balance index must, by definition, attain its extremal values only on trees that can reasonably be described as extremely balanced or extremely imbalanced. Also, knowing the tree types for which an index attains its maximum and minimum values aids interpretation of the index, by providing reference points to which other trees can be compared.
Given these motivations,} we have focussed on leafy broom trees, which provide us with a relatively simple yet non-trivial special case for investigating minimal $J^1$ values. Broom trees include caterpillar trees, which are the least-balanced bifurcating trees according to popular conventional tree balance indices \citep{fischer_extremal_2021, coronado2020minimum}.

{
We have thus shown that, although caterpillar trees minimize $J^1$ among leafy bifurcating trees with equally sized leaves, they are sub-optimal when outdegrees greater than two are permitted.
The least balanced broom tree in the more general case has a long handle and a large head (containing at least half the leaves). Nevertheless, most broom trees (all but those with extremely large heads) have relatively similar $J^1$ values, which decrease rapidly as the number of leaves increases.
}
%The reason for this difference between the $J^1$ index and the Sackin index is that $J^1$ accounts not only for leaf depths but also for per-node imbalance.
%\yan{What is the intuition? What is taken into account here and not by previous indixes that lead to this difference?}
%\arm{If you compare to Sackin index, that index is said to minimize a tree that maximizes total path length to leaves. the J1 index kind of does that but it adds a further weight trying to minimize subtree balance (subtrees rooted at every internal node). 
{
By relaxing the assumption of equal leaf sizes, we have further shown that a small change in relative leaf sizes can dramatically change the topology of the least balanced broom tree. In particular, if the leaves in the broom head are moderately smaller than those in the handle then the least balanced broom tree is a caterpillar.
}
%For example, at the boundary of regions $R_3$ and $R_4$ in Figure~\ref{boundaries_plot}, the least balanced tree switches from being a caterpillar to being a tree of height two.

{
Is it at all problematic that caterpillar trees do not always minimize $J^1$? Does this property make $J^1$ somehow different from other tree balance indices?
We argue that the answer is no in both cases.
It is true that caterpillar trees minimize Sackin's index among all trees with a given number of leaves and no nodes of outdegree one \citep{fischer_extremal_2021}. But this property is meaningless because it is meaningless to compare Sackin's index values for trees with different outdegree distributions \citep{Shao1990, lemant_robust_2022}. Otherwise, for example, the symmetrical bifurcating tree on four leaves would be considered less balanced than the four-leaved star tree. The same restriction applies to the total cophenetic index \citep{mir_new_2013}. The other popular conventional tree balance index, Colless' index, is defined only for bifurcating trees \citep{Colless1982}. Hence there is no prior ``correct" answer as to whether a given non-bifurcating broom tree is more or less balanced than the caterpillar with the same number of leaves. Our index provides a reasonable, consistent solution to this problem.}

An important goal for future research is to obtain the expected values of $J^1$ for more complicated yet more biologically realistic models. These include models that generate non-leafy trees, such as the clone trees that are central to cancer evolution research \citep{noble2022spatial, feder2024detecting}. Our recent generalization of $J^1$ to account for branch lengths \citep{noble_new_2023} motivates the study of models that generate non-uniform branch lengths. It should also be possible to investigate the expected values of extensions of $J^1$ that apply to phylogenetic networks, as was recently done for Sackin's index \citep{fuchs2025sackin}.

The methods we have applied here to finding least-balanced broom trees can be readily adapted to other special tree topologies. In particular, given that Conjecture~\ref{conj_brooms} applies only in the case of equally-sized leaves, it will be worthwhile to examine minimal $J^1$ values when the node size distribution is such that broom trees clearly do not minimize $J^1$. For example, every broom tree in region $R_4$ of Figure~\ref{boundaries_plot} has a higher $J^1$ value than the leafy star tree with the same number of leaves and the same leaf size distribution.
%Other important open problems include our conjecture that the least balanced leafy tree with uniform leaf sizes and no nodes of outdegree one is a broom tree.

In conclusion, our results strengthen the case for $J^1$ as the most useful cross-disciplinary index of tree balance and provide a firm foundation for using $J^1$ instead of conventional indices to compare and categorize empirical trees.

\section*{Data availability statement}

All data created for this study can be readily reproduced using our mathematical methods, with the exception of the data presented in Tables~\ref{table_values} and ~\ref{table_values_unif}. Efficient R code for calculating $J^1$ is available at https://zenodo.org/records/5873857.

\section*{Acknowledgements}
This work was supported by a London Mathematical Society Research in Pairs award (reference 42320) to R.N. and Y.V. and by funding from the National Cancer Institute of the National Institutes of Health under Award Number U54CA217376 to R.N. and V.M. The content is solely the responsibility of the authors and does not necessarily represent the official views of the National Institutes of Health.
The authors thank Panos Giannopoulos for helpful discussions about tree balance in computer science and Holden Lee for advice on methods of approximating the inverse Sackin index.
	
\bibliography{references}

%\printbibliography

\begin{appendices}
	\section{Additional proofs and derivations}\label{proof-app}

    \subsection{Proof of Proposition~\ref{jensen_prop}}\label{jensen_proof}

    \begin{proof}
 
	Our proof of Proposition~\ref{jensen_prop} relies on a recent result of \citet{liao_sharpening_2017}.
	\begin{theorem}[\citet{liao_sharpening_2017}]\label{jensen_thm}
		Let $X$ be a one-dimensional random variable with mean $\mu$, and $P(X\in(a,b))=1$, where $\infty\leq a<b\leq \infty$. If $f(x)$ is a twice differentiable function on $(a,b)$, and
		\begin{equation*}
			h(x;\nu) = \frac{f(x)-f(\nu)}{(x-\nu)^2} - \frac{f'(\nu)}{x-\nu},
		\end{equation*}
		then
		\begin{equation}
			\inf_{x\in(a,b)}\{h(x;\mu)\}\mathrm{Var}(X) \leq \mathbb{E}[f(X)] - f(\mathbb{E}[X]) \leq \sup_{x\in(a,b)}\{ h(x;\mu) \}\mathrm{Var}(X).
		\end{equation}
	\end{theorem}
	%We are interested in the case $f(x) = 1/x$.

	\subsubsection*{Part i}
 
		Let $\mu_Y$ be the expected value of the Sackin index under the Yule process for trees on $n$ leaves, let $f(x)=\frac{n \log_2 n}{x}$, and define
		\begin{equation}\label{hxmuY}
			h(x;\mu_Y) = \frac{f(x)-f(\mu_Y)}{(x-\mu_Y)^2} - \frac{f'(\mu_Y)}{x-\mu_Y} = \frac{n \log_2 n}{x\mu_Y^2}.
		\end{equation}
		Theorem~\ref{jensen_thm} then implies
		\begin{equation}\label{JensenApp}
			\frac{n\log_2n}{\frac{(n-1)(n+2)}{2}\mu_Y^2} \text{Var}_Y(I_S) \leq \mathbb{E}[J^1] - \frac{n\log_2n}{\mathbb{E}[I_S]} \leq \frac{n\log_2n}{\mu_Y^2n\log_2n}\text{Var}_Y(I_S),
		\end{equation}
		where the supremum and infimum of $h(x,\mu)$ are the extremal values of the Sackin index on bifurcating trees \citep{fischer_extremal_2021}. The expectation of the Sackin index under the Yule process is given in Equation~\ref{yule_exp_sackin}, and its variance as derived by \citet{cardona_exact_2012} is
		\begin{equation}\label{yulevarIS}
			\text{Var}_Y(I_S) = 7n^2 - 4n^2\sum_{i=1}^n\frac{1}{i^2} - 2n\sum_{i=1}^{n}\frac{1}{i} - n.
		\end{equation}
		Substituting these expressions into Equation~\ref{JensenApp}, we obtain the limits
		\begin{align*}
			\frac{n\log_2n}{\frac{(n-1)(n+2)}{2}\mu_Y^2} \text{Var}_Y(I_S) & \stackrel{n\to\infty}{\sim} \frac{\log_2 n\left(7n^2 - 4n^2\sum_{i=1}^n\frac{1}{i^2}-2n\sum_{i=1}^n\frac{1}{i}-n\right)}{2n^3\left(\sum_{i=2}^n\frac{1}{i}\right)^2} \\
		& \stackrel{n\to\infty}{\sim} \frac{\left(7-\frac{2\pi^2}{3}\right)n^2\log_2 n}{2n^3(\ln n)^2}\\
        & \stackrel{n\to\infty}{\sim} \frac{21-2\pi^2}{6n \ln 2 \ln n}
        \to 0
		\end{align*}
		for the lower bound on the gap, and
		\begin{align*}
			\frac{n\log_2n}{\mu_Y^2n\log_2n}\text{Var}_Y(I_S) & \stackrel{n\to\infty}{\sim} \frac{7n^2 - 4n^2\sum_{i=2}^n\frac{1}{i^2}-2n\sum_{i=2}^n\frac{1}{i}-n}{4n^2\left(\sum_{i=2}^n\frac{1}{i}\right)^2} \\
   & \stackrel{n\to\infty}{\sim} \frac{\left(7-4\left(\frac{\pi^2}{6}-1\right)\right)n^2}{4n^2(\ln n)^2}\\
   & \stackrel{n\to\infty}{\sim} \frac{7-4\left(\frac{\pi^2}{6}-1\right)}{4(\ln n)^2}\to 0  
		\end{align*}
		for the upper bound on the gap. The upper bound reaches a maximum of approximately $0.00790$ at $n=13$ and the lower bound reaches a maximum of approximately $0.00499$ at $n=8$.
	
\subsubsection*{Part ii}
 
		Let $\mu_U$ be the expected value of the Sackin's index under the uniform model for trees on $n$ leaves, and $f(x)=\frac{n \log_2 n}{x}$, and define
		\begin{equation}\label{hxmuU}
			h(x;\mu_U) = \frac{f(x)-f(\mu_U)}{(x-\mu_U)^2} - \frac{f'(\mu_U)}{x-\mu_U} = \frac{n \log_2 n}{x\mu_U^2}.
		\end{equation}
		Theorem~\ref{jensen_thm} then implies
		\begin{equation}\label{JensenAppU}
			\frac{n\log_2n}{\frac{(n-1)(n+2)}{2}\mu_U^2} \text{Var}_U(I_S) \leq \mathbb{E}[J^1] - \frac{n\log_2n}{\mathbb{E}[I_S]} \leq \frac{n\log_2n}{\mu_U^2n\log_2n}\text{Var}_U(I_S),
		\end{equation}
		analogously to Equation~\ref{JensenApp}. The expectation and variance of Sackin's index under the uniform model are \citep{cardona_exact_2012}:
		\begin{gather}
			\mathbb{E}_U(I_S) = \frac{4^{n-1}n!(n-1)!}{(2n-2)!}-n, \\
			\mathbb{V}_U(I_S) = n\frac{10n^2 - 3n - 1}{3} - \frac{(n+1)n}{2}\frac{(2n-2)!!}{(2n-3)!!}-n^2\left(\frac{(2n-2)!!}{(2n-3)!!}\right)^2.
		\end{gather}
		In the limit $n\to\infty$, we obtain:
		%\begin{align}
		%	n! \stackrel{n\to\infty}{\sim}              & \sqrt{2\pi n}\left(\frac{n}{e}\right)^n \label{stirling} \\
		%	n!! \stackrel{n\to\infty}{\sim}             & \begin{cases}                                            
		%	\sqrt{\pi n}\left(\frac{n}{e}\right)^{n/2}, & \text{\quad $n$ even,}                                   \\
		%	\sqrt{2n}\left(\frac{n}{e}\right)^{n/2},    & \text{\quad $n$ odd,}                                    
		%	\end{cases}
		%\end{align}
		%to obtain the asymptotic behaviour of the expected value and variance of $I_S$ under the uniform model. The expectation and variance reduce to 
		%\begin{align*}
		%	\mathbb{E}_U(I_S)\stackrel{n\to\infty}{\sim} & \frac{4^{n-1}\sqrt{2\pi n}\left(\frac{n}{e}\right)^n \sqrt{2\pi (n-1)}\left(\frac{n-1}{e}\right)^{n-1}}{\sqrt{2\pi (2n-2)}\left( \frac{2n-2}{e} \right)^{2n-2}} - n \\
		%	\stackrel{n\to\infty}{\sim}                                        & \frac{\sqrt{\pi n}n^n}{e(n-1)^{n-1}}-n \\
		%	\stackrel{n\to\infty}{\sim}                                        & e^{-1}\sqrt{\pi} \exp\left[ \left(n+\frac{1}{2}\right)\log n - (n-1)\log(n-1) \right] - n \\
		%	\stackrel{n\to\infty}{\sim}                                        & e^{-1}\sqrt{\pi} n^{\frac{1}{2}}(n-1)\exp\left[n\log \frac{n}{n-1}\right] - n\\
        %    \stackrel{n\to\infty}{\sim}                                        &
        %    \sqrt{\pi}n^{\frac{3}{2}}
        %    \end{align*}

        \begin{align}
            \mathbb{E}_U(I_S)\stackrel{n\to\infty}{\sim} & \sqrt{\pi}n^{3/2}\nonumber\\
            \mathbb{V}_U(I_S)\stackrel{n\to\infty}{\sim} & \left(\frac{10}{3}-\pi\right)n^3.\nonumber
        \end{align}
        
		%and the variance simplifies to
		%\begin{align*}
		%\begin{split}	     
        %\text{Var}_U(I_S)  \stackrel{n\to\infty}{\sim} &  \frac{10}{3}n^3 - n^2 - \frac{1}{3}n - \frac{n^2 + 3n +2}{2}\frac{\sqrt{\pi (2n-2)}\left( \frac{2n-2}{e}^{n-1} \right)}{\sqrt{2 (2n-3)}\left( \frac{2n-3}{e}^{n-1/2} \right)} \\
		%	     & -  n^2\left( \frac{\sqrt{\pi (2n-2)}\left( \frac{2n-2}{e}^{n-1} \right)}{\sqrt{2 (2n-3)}\left( \frac{2n-3}{e}^{n-1/2} \right)} \right)^2
        %\end{split}
        %\\ 
        %\begin{split}			
   %\stackrel{n\to\infty}{\sim} & \frac{10}{3}n^3 - n^2 - \frac{1}{3}n \\
   %& - \frac{n^2+3n+2}{2}\sqrt{\frac{e\pi}{2}}\exp\left[(n-1)\log\frac{2n-2}{2n-3}+\frac{1}{2}\log(2n-3)\right]\\
		%	     & - n^2\exp\left[ \log(2n-3) \right]
        %\end{split}
        %\\
		%	\stackrel{n\to\infty}{\sim} & \frac{4}{3}n^3 + 2n^2 - \frac{1}{3}n - \frac{n^\frac{5}{2}+3n^\frac{3}{2}+2n^\frac{1}{2}}{2}\sqrt{e\pi}\\
        %    \stackrel{n\to\infty}{\sim} & \frac{4}{3}n^3
		%\end{align*}
  
  It follows that the lower bound of $\mathbb{E}_U(J^1)-\frac{n\log_2n}{\mathbb{E}_U(I_S)}$ goes to $0$ as fast as $\frac{\ln n}{n}$, while the upper bound increases very slowly to the limit $\frac{10}{3\pi}-1 \approx 0.061$ as $n\to\infty$, as a consequence of high variance in the uniform model (Figure~\ref{var_fig}). 
	\end{proof}

	\begin{figure}
		\centering
		\includegraphics[width=0.65\textwidth]{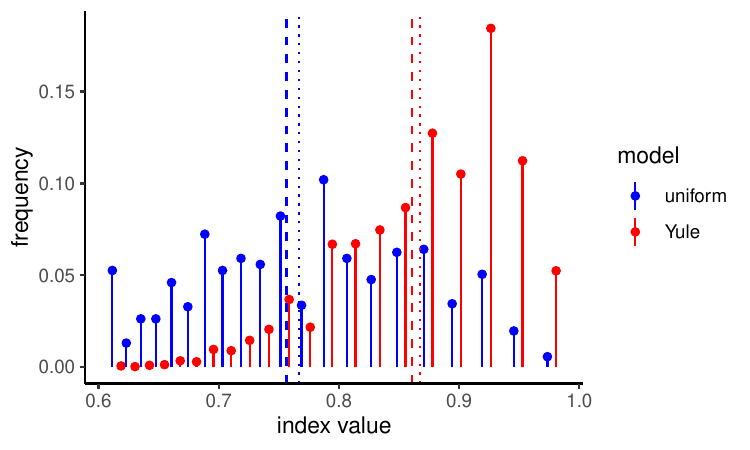}
		\caption{Exact distributions of $J^1$ values for $10$-leaf trees generated under the Yule process (red) and the uniform model (blue), illustrating that the latter has higher variance. Dashed lines represent $\mathbb E (J^1)$ and dotted lines $\frac{n\log_2n}{\mathbb{E}(I_S)}$.}
		\label{var_fig}
	\end{figure}

	% \begin{figure}
	% 	\centering
	% 	\includegraphics[width=0.7\textwidth]{figures/jensen_gap_1k.pdf}
	% 	\caption{\textcolor{red}{NEED TO MERGE THIS WITH FIGURE~\ref{jensenfig}.} The convergence of the upper bound to $4/3\pi$ is much slower than the convergence of the lower bound to $0$, and the maximum it reaches over the plotted range is $0.0604$ for $n=128000$. The red crosses, as in figure \ref{jensenfig}, suggest convergence of the gap size.}
	% 	\label{conv_upper_fig}
	% \end{figure}

   \subsection{Approximations of $\mathbb{E}[J^1]$}\label{app:approx_j1}

    Let $f(x)=\frac{1}{x}$. We are interested in an approximation of such a function in terms of polynomial-order terms in $x$, and hence it is natural to consider its $k$th-order Taylor approximation about $\mu>0$:
%\[f(x)=\frac{1}{\mu}-\frac{1}{\mu^2}(x-\mu)+\frac{1}{\mu^3}(x-\mu)^2+R_2(x)\]
\[f(x)=\sum_{i=0}^k (-1)^{i}\frac{1}{\mu^{i+1}}(x-\mu)^i+R_k(x).\]
Here $R_k(x)$ is the error:

\begin{align}
    R_k(x)=\frac{f^{(k+1)}(z(x))}{(k+1)!}(x-\mu)^{k+1}=(-1)^{k+1}\frac{(x-\mu)^{k+1}}{z^{k+2}(x)}\label{eq:error}
\end{align}
where $z(x)$ is a value between $x$ and $\mu$. Composing $f$ with the random variable $I_S$ representing the Sackin index of a randomly generated tree with $n$ leaves, letting $\mu=\mathbb{E}[I_S]$, and taking expectations, we obtain
\begin{align}
\mathbb{E}\left[\frac{J^1}{n\log_2 n}\right]=\mathbb{E}[f(I_S)]=\frac{1}{\mathbb{E}[I_S]}+\sum_{i=2}^k (-1)^{i}\frac{1}{\mu^{i+1}}\mathbb{E}[(I_S-\mu)^i]+\mathbb{E}[R_k(I_S)],\label{ref:taylor_general}
\end{align}
%We want to control the magnitude of the last term
%\[|\mathbb{E}\left[\frac{(I_S-\mu)^3}{z(I_S)^4}\right]|\leq \mathbb{E}\left[|\frac{(I_S-\mu)^3}{z(I_S)^4}|\right]\]
where the first equality is due to the leafy-tree identity.
Thus, the goal is to examine the asymptotic behavior of $\mathbb{E}[R_k(I_S)]$ as $n$ goes to infinity. To generate random trees, we can examine either the Yule or the uniform process. We denote $I_{S,k}$ the random variable corresponding to the Sackin index of either a Yule or uniform tree with $k$ leaves. The subscript $k$ is omitted when the number of leaves is clear from context.
%\yan{Armaan, why do we say that we want to examine the behavior of $|R_k(I_S)|$ while we are actually interested in $\mathbb{E}[R_k(I_S)]$, without absolute value? Because the former leads to the latter?}

\subsubsection{Yule}\label{sec:yule}
%Let $I_{S,k}$ denote the \yan{random variable corresponding to the} Sackin index of a Yule tree with $k$ leaves. 
%Under the Yule process, the Sackin index 
Under the Yule process, $I_S$ satisfies the following recurrence \citep{blum2005_statistic}:
\begin{align}
I_{S,n}\stackrel{d}{=}I_{S,V}+I_{S,n-V}+n,
\label{eq:rec}\end{align}
where $V$ is a discrete uniform random variable with support $\{1,\dots,n-1\}$, %and 
$I_{S,V},I_{S,n-V}$ are independent, conditional on $V$, and $\stackrel{d}{=}$ denotes equality in distribution. In fact it is easy to generate simulations of $I_S$ using Equation~\ref{eq:rec}.

It is known that the distribution of $\frac{I_{S}-\mathbb{E}[I_{S}]}{n}$ is the same as that of $\frac{Q_{n}-\mathbb{E}[Q_{n}]}{n}$, where $Q_{n}$ is the number of steps required for a Quicksort of a random $n$-length string. This follows from $I_{S}\stackrel{d}{=}Q_{n}+2(n-1)$ \citep{iliopoulos2015}. Additionally, $\frac{Q_{n}-\mathbb{E}[Q_{n}]}{n}$ converges in distribution to a random variable $Y$, which is %(uniquely) 
characterized by \citet{rosler1991_quicksort}.
%\[Y\stackrel{d}{=}Y\cdot U+Y'\cdot (1-U)+2U\ln U+2(1-U)\ln (1-U)+1,\]
%where $Y,Y'$ are i.i.d., and $U$ is independent and uniform on $[0,1]$. \yan{Armaan, do we ever use this formula? If not, why do we care?}
An even stronger convergence is true, where (absolute) moments converge to (absolute) moments of this limiting distribution (Wasserstein $d_p$ convergence, for any $p$) \citep{rosler1991_quicksort}. Thus: 
%Hence, this holds for $\frac{I_{S,n}-\mathbb{E}[I_{S,n}]}{n}$, and thus
%\yan{Thus, dropping from now on the subscript $n$ in $I_{S, n}$:}
\begin{align}
    \mathbb{E}[|I_S-\mathbb{E}[I_S]|^p]\sim \mathbb{E}[|Y|^p]n^p,\nonumber\\
    \mathbb{E}[(I_S-\mathbb{E}[I_S])^p]\sim \mathbb{E}[Y^p]n^p.\label{eq:centered_moment_yule}
\end{align}

Furthermore, recalling that the expectation of $I_S$ is equivalent to $2 n \ln n$ and from rearranging (using continuous mapping) we note that:%\yan{Armaan, details would help me for the first implication, though probably because of my lack of familiarity with standard manipulations in probability theory. Also, we should probably recall Eq. 8, that is, that the expectation of $I_S$ is equivalent to $2 n \ln n$, a fact which seems used several times.}
\begin{align}
    \frac{I_S-\mathbb{E}[I_S]}{n}\stackrel{d}{\sim} Y\Rightarrow \frac{I_S}{\mathbb{E}[I_S]}\stackrel{p}{\sim} 1\Rightarrow J^1=\frac{n\log_2 n}{I_S}\stackrel{p}{\sim} \frac{1}{2\ln 2}.\label{eq:asymp_yule}
\end{align}
So, asymptotically, $J^1\approx \frac{1}{2\ln 2} \approx 0.72$  (in probability, which we also see from simulations).

Now, we characterize the error term associated with the 2nd order Taylor approximation ($\mathbb{E}[R_k(I_S)]$ for $k=2$). Using (\ref{eq:error}), the error can be written as $\mathbb{E}\left[\frac{(I_S-\mathbb{E}[I_S])^3}{z(I_S)^4}\right]$. %Recall \yan{Armaan, recall from where? I understand from Eqs 8 and 9 that this is the approximation we propose for the expectation of $J_1$ and Appendix A2 is about the expectation of $J_1$, but this should be clarified.} we seek $\mathbb{E}[n\log_2 n / I_S]$, which we can obtain with the Taylor approximation, and the error associated with a 2nd-order approximation is, again $\mathbb{E}\left[\frac{(I_S-\mathbb{E}[I_S])^3}{z(I_S)^4}\right]$.
Since $z(I_s)\in [\min\{\mathbb{E}[I_S],I_S\},\max\{\mathbb{E}[I_S],I_S\}]$ and by continuous mapping $\frac{\min\{\mathbb{E}[I_S],I_S\}}{\mathbb{E}[I_S]}\stackrel{p}{\sim}1,\frac{\max\{\mathbb{E}[I_S],I_S\}}{\mathbb{E}[I_S]}\stackrel{p}{\sim}1$ (using (\ref{eq:asymp_yule})), so $z(I_S) / \mathbb{E}[I_S]\stackrel{p}{\sim} 1$. This motivates (but does not yet show) the following asymptotic approximation to the error term:
\[\frac{\mathbb{E}\left[\frac{(I_S-\mathbb{E}[I_S])^3}{z(I_S)^4}\right]}{\mathbb{E}\left[\frac{(I_S-\mathbb{E}[I_S])^3}{\mathbb{E}[I_S]^4}\right]}\to 1\]

Showing the above is equivalent to showing:
\[\frac{\mathbb{E}\left[(I_S-\mathbb{E}[I_S])^3\left(\frac{\mathbb{E}[I_S]^4}{z(I_S)^4}-1\right)\right]}{\mathbb{E}[(I_S-\mathbb{E}[I_S])^3]}\to 0.\]

Using Cauchy-Schwarz:
\begin{align*}
\left|\frac{\mathbb{E}\left[(I_S-\mathbb{E}[I_S])^3\left(\frac{\mathbb{E}[I_S]^4}{z(I_S)^4}-1\right)\right]}{\mathbb{E}[(I_S-\mathbb{E}[I_S])^3]}\right|
&\leq \frac{\sqrt{\mathbb{E}\left[(I_S-\mathbb{E}[I_S])^6\right]\mathbb{E}\left[\left(\frac{\mathbb{E}[I_S]^4}{z(I_S)^4}-1\right)^2\right]}}{|\mathbb{E}[(I_S-\mathbb{E}[I_S])^3]|} \\
&\sim \frac{\sqrt{\mathbb{E}[Y^6]}}{|\mathbb{E}[Y^3]|}\sqrt{\mathbb{E}\left[\left(\frac{\mathbb{E}[I_S]^4}{z(I_S)^4}-1\right)^2\right]}.
\end{align*}

It then suffices to show $\mathbb{E}\left[\left(\frac{\mathbb{E}[I_S]^4}{z(I_S)^4}-1\right)^2\right]\to 0$. Define the set $$
A=\{\omega:I_S(\omega)\geq \mathbb{E}[I_S]/2\}.
$$
and let $A^c$ denote its complement. Noting that:
\[\mathbb{E}\left[\left(\frac{\mathbb{E}[I_S]^4}{z(I_S)^4}-1\right)^2\right]=\mathbb{E}\left[\left(\frac{\mathbb{E}[I_S]^4}{z(I_S)^4}-1\right)^2\mathbf1_A\right]+\mathbb{E}\left[\left(\frac{\mathbb{E}[I_S]^4}{z(I_S)^4}-1\right)^2\mathbf{1}_{A^c}\right]\]
it suffices to prove that each summand converges to 0. The %general strategy will be 
strategy to control the first term is to show that on the event $A$, the expectation decays to $0$ as $n\to\infty$. To control the second term, we show that the probability of $A^c$ decays to $0$ sufficiently quickly as $n\to\infty$.  

On set $A$, $\frac{\mathbb{E}[I_S]}{z(I_S)}\leq 2$, so
\begin{align*}
\left|\frac{\mathbb{E}[I_S]^4}{z(I_S)^4}-1\right| 
&= \left|\frac{\mathbb{E}[I_S]} {z(I_S)}-1\right|\left|\left(\frac{\mathbb{E}[I_S]}{z(I_S)}\right)^3+\left(\frac{\mathbb{E}[I_S]}{z(I_S)}\right)^2+\frac{\mathbb{E}[I_S]}{z(I_S)}+1\right| \\
&\leq 15\left|\frac{\mathbb{E}[I_S]}{z(I_S)}-1\right| 
=15\left|\frac{\mathbb{E}[I_S]-z(I_S)}{z(I_S)}\right|\\
&\leq 15\left|\frac{\mathbb{E}[I_S]-I_S}{z(I_S)}\right|\leq 30\frac{|I_S-\mathbb{E}[I_S]|}{\mathbb{E}[I_S]}.
\end{align*}
%\arm{Should be more clear, put intermediate step}\yan{Armaan, about the last inequality, you are too modest: if you use naively the inequality $z(I_S) \geq\mathbb{E}[I_S]/2$, you just get $\left|\frac{\mathbb{E}[I_S]}{z(I_S)}-1\right|\leq 1$, so you must do something more sophisticated!}
Thus:
\[\mathbb{E}\left[\left(\frac{\mathbb{E}[I_S]^4}{z(I_S)^4}-1\right)^2\cdot \mathbf{1}_{A}\right]\leq \frac{900}{\mathbb{E}[I_S]^2}\mathbb{E}\left[|I_S-\mathbb{E}[I_S]|^2\right]= \mathcal{O}\left(\frac{1}{(\ln n)^2}\right)\to 0.\]

On the complement of set $A$ (that is, on $A^c$), $I_S \leq \mathbb{E}[I_S]/2 \leq \mathbb{E}[I_S]$ so $z(I_S) \leq \mathbb{E}[I_S]$. Moreover, due to the trivial bound $I_S\geq n$, we have $z(I_S)\geq n$. Hence
\[\left|\frac{\mathbb{E}[I_S]^4}{z(I_S)^4}-1\right| = \frac{\mathbb{E}[I_S]^4}{z(I_S)^4}-1  \leq \frac{\mathbb{E}[I_S]^4}{z(I_S)^4}\leq \widetilde{C}(\ln n)^4\]
for some constant $\widetilde{C}$. Notice further that
\begin{align*}
\mathbb{P}\left(I_S<\frac{\mathbb{E}[I_S]}{2}\right)
&\leq \mathbb{P}\left(|I_S-\mathbb{E} [I_S]|>\frac{\mathbb{E}[I_S]}{2}\right) 
=\mathbb{P}\left(|I_S-\mathbb{E}[I_S]|^9>\left(\frac{\mathbb{E}[I_S]}{2}\right)^9\right) \\
&\leq \frac{2^9\mathbb{E}[|I_S-\mathbb{E}[I_S]|^9]}{\mathbb{E}[I_S]^9}
= \mathcal{O}\left(\frac{1}{(\ln n)^9}\right).
\end{align*}
Hence
\begin{align*}
\mathbb{E}\left[\left(\frac{\mathbb{E}[I_S]^4}{z(I_S)^4}-1\right)^2\cdot \mathbf{1}_{A^c}\right]
&\leq \widetilde{C}^2(\ln n)^8\mathbb{P}(A^c) \\
&\leq \frac{2^9\mathbb{E}[|I_S-\mathbb{E}[I_S]|^9]}{\mathbb{E}[I_S]^9}\widetilde{C}^2(\ln n)^8
= \mathcal{O}\left(\frac{1}{\ln n}\right)
\to 0.
\end{align*}

Finally then
\[
\mathbb{E}\left[\left(\frac{\mathbb{E}[I_S]^4}{z(I_S)^4}-1\right)^2\right]
=\mathbb{E}\left[\left(\frac{\mathbb{E}[I_S]^4}{z(I_S)^4}-1\right)^2\mathbf{1}_A\right]+\mathbb{E}\left[\left(\frac{\mathbb{E}[I_S]^4}{z(I_S)^4}-1\right)^2\mathbf{1}_{A^c}\right]\to 0.
\]

So, we can indeed say that
\begin{align}
    \mathbb{E}\left[\frac{(I_S-\mathbb{E}[I_S])^3}{z(I_S)^4}\right]\sim \mathbb{E}\left[\frac{(I_S-\mathbb{E}[I_S])^3}{\mathbb{E}[I_S]^4}\right].\label{eq:replace_error}
\end{align}

Thus using Eqs. \ref{eq:replace_error} and \ref{eq:centered_moment_yule}:

\[\mathbb{E}[n\log_2 n\cdot R_2(I_S)] \sim\frac{n\log_2 n}{(2n\ln n)^4}\mathbb{E}[Y^3]n^3=\frac{\mathbb{E}[Y^3]}{16\ln 2}\frac{1}{(\ln n)^3}.\]

Expressions for moments of $Y$ are given, in a rather cumbersome recursive form, by \citet{rosler1991_quicksort}. In general, for a $k$th order Taylor approximation, following the same steps as above, 
$$
\mathbb{E}[(n\log_2 n)R_k(I_S)] %\approx
\sim \frac{\mathbb{E}[Y^{k+1}]}{2^{k+2}\ln 2(\ln n)^{k+1}}.
$$
Since $\mathbb{E}[Y^2]$ is simply the leading term coefficient of $\mathbb{V}[I_S]$, we verify in Figure \ref{jensenfig} that the theoretical error for the first-order approximation is close to the empirical error.

%Now, the asymptotic error term is reasonable to obtain for $k=1$. Thus, instead of performing a 2nd order Taylor expansion, we can add the above asymptotic error term: 

%\[\mathbb{E}[J^1]\approx \frac{n\log_2 n}{\mathbb{E}[I_S]}+\frac{\mathbb{E}[Y^2]}{8\cdot \ln 2 (\ln n)^2}\]

%Empirically, the error decays even quicker than for a second-order Taylor approximation.

\subsubsection*{Full $m$-ary generalized Yule trees}\label{Yule_m-ary}

Although $m$-ary trees are less prevalent than bifurcating trees in the evolutionary literature, we note that the above procedure can still be easily generalized as follows.

Consider a generalized Yule process in which a leaf is selected uniformly at random and replaced by a star subtree of outdegree $m$. For general $m$-ary trees generated under this generalized Yule process, we have:
\begin{align*}
    &\mathbb{E}[I_S] = (1+\theta)n\left(\gamma(i+\theta+1)-\gamma(\theta+1)\right)\\
    &\mathbb{V}[I_S] =n^2\left((2+\theta)\frac{i}{i+\theta}-(1+\theta)^2(\gamma'(\theta+1)-\gamma'(i+\theta+1))- \right. \\
    &\qquad\qquad\qquad\left. \frac{1}{i+\theta}(\gamma(i+\theta+1)-\gamma(\theta+1))\right)
\end{align*}

where $\theta=\frac{1}{m-1}$, $i$ is the number of internal nodes, which satisfies $n=(m-1)i+1$ and hence $i = (n-1)/(m-1)$, and $\gamma$ is the digamma function 
\citep{halton1989}. %\yan{Armaan, why do we use this notation for derivatives even though we just use the order $0$ and $1$?}
Thus, $\mathbb{E}[I_S]\sim (1+\theta)\cdot n\ln n$ and $\mathbb{V}[I_S]\sim C_2\cdot n^2$. Therefore:
\[
\mathbb{P}\left(\left|\frac{I_S}{\mathbb{E}[I_S]}-1\right|>\varepsilon\right)=\mathbb{P}(|I_S-\mathbb{E}[I_S]|>\mathbb{E}[I_S]\varepsilon)\leq \frac{\mathbb{V}[I_S]}{(\mathbb{E}[I_S]\varepsilon)^2}\rightarrow 0.
\]
So $\frac{I_S}{\mathbb{E}[I_S]}\stackrel{p}{\sim} 1$. Using the leafy tree identity (Proposition \ref{proposition6}) and continuous mapping, 
%$\frac{I_S}{\mathbb{E}[I_S]}\stackrel{p}{\sim} 1\Rightarrow$% 
this implies that $J^1\stackrel{p}{\sim} \frac{1}{(1+\theta)\ln m}$. Thus, we may apply similar analysis to above (Section \ref{sec:yule}) to show that a $k$th order Taylor approximation's error would decay at rate $1/(\ln n)^{(k+1)}$.

\subsubsection{Uniform}

Under this model, each tree of $n$ leaves is as likely as any other. It is important to note that the uniform model selects higher Sackin-index trees than the Yule model. 

%\yan{Armaan, let us discuss if we need to add: Let $I_{S, n}$ denote the random variable corresponding to the Sackin index of a tree with $n$ leaves under the uniform model.}
Let $X_n=\frac{I_{S,n}}{n^{3/2}}$. Under the uniform model, $\lim_{n\rightarrow\infty}\mathbb{E}[X_n^r]= \mathbb{E}[A^r]$ \citep{takacs_uniform}, where $A$ is the Airy distribution and $r\in \mathbb{Z}^+$. Such moments are quite simple to obtain:
\[\mathbb{E}[A^r]=\sqrt{8}^r n^{3r/2}M_r\]
where
\[M_r=K_r\cdot \frac{4\sqrt{\pi}r!}{\Gamma\left(\frac{3r-1}{2}\right)2^{r/2}},K_r=\frac{3r-4}{4}K_{r-1}+\sum_{j=1}^{r-1} K_jK_{r-j}\]
and $K_0=-\frac{1}{2}$ \citep{takacs_uniform}.

We have exact expressions for the mean and variance of $I_S$ under the uniform model, but the above can be used to generate $k$th-order Taylor approximations of $\mathbb{E}[J^1]$ for $k\geq 3$ using asymptotic approximations of moments if desired. Further, note that $\mathbb{E}[A^{-1}]=\frac{\sqrt{\pi}}{4}\left(3-16\frac{\pi^3\sqrt{3}}{\Gamma(1/3)^6}\right)$ \citep{flajolet2001}. Thus, conjecturing on the convergence of the inverse moment $\mathbb{E}[X_n^{-1}]\rightarrow \mathbb{E}[A^{-1}]$, we obtain the following large $n$ approximation:
\[\mathbb{E}[J^1]\sim \mathbb{E}[A^{-1}]n^{-1/2}\log_2 n,\]
which is shown in Figure $\ref{jensenfig}$.

% \iffalse
% \[\frac{1}{x}=\sum_{i=0}^\infty (-1)^i\frac{1}{\mu^{i+1}}(x-\mu)^i\]

% And we substitute $x$ with $X_n$ and take expectations:

% \[\mathbb{E}[1/X_n]=\mathbb{E}\left[\sum_{i=0}^\infty (-1)^i\frac{1}{\mathbb{E}[X_n]^{i+1}}(X_n-\mathbb{E}[X_n])^i\right]\]

% \textbf{IF} we can interchange expectation and infinite sum, then we have

% \[\mathbb{E}[1/X_n]=\sum_{i=0}^\infty (-1)^i\frac{1}{\mathbb{E}[X_n]^{i+1}}\mathbb{E}\left[(X_n-\mathbb{E}[X_n])^i\right]\]

% We should take limits on both sides, to attain:

% \[\lim_{n\rightarrow\infty }\mathbb{E}[1/X_n]=\lim_{n\rightarrow\infty}\lim_{k\rightarrow\infty}\sum_{i=0}^k (-1)^i\frac{1}{\mathbb{E}[X_n]^{i+1}}\mathbb{E}\left[(X_n-\mathbb{E}[X_n])^i\right]\]

% \textbf{If} we can interchange the two limits, then:

% \[\lim_{n\rightarrow\infty }\mathbb{E}[1/X_n]=\lim_{k\rightarrow\infty}\sum_{i=0}^k (-1)^i\frac{1}{\mathbb{E}[A]^{i+1}}\mathbb{E}\left[(A-\mathbb{E}[A])^i\right]=\mathbb{E}[A^{-1}]\]

% Where the second equality is true if we can again take the expectation outside the limit. And, if all this is true, we have what we want to show.

% (It may be important to note that the limiting moments behave as follows: $\mathbb{E}[A^r]\sim \sqrt{2}\cdot 3\cdot r\left(\frac{2r}{3e}\right)^{r/2}$)\fi

\subsection{Derivation of Equation~\ref{J1Tb}}\label{J1Tb_proof}
 
        For brevity we present the derivation in the case of equally sized leaves. The derivation for general $p > 0$ is similar.
		\begin{align*}
			J^1_B(n, k, 1) & = \frac{1}{\sum_{l=k}^{n}l} \sum_{i\in\tilde{V}} S_i^{*}\sum_{j\in C(i)}W_{ij}^{1}                                             \\
			         & = \frac{-2}{(n+k)(n-k+1)}\sum_{i \in \tilde{V}} S_i^*\sum_{j\in C(i)} \frac{S_j}{S_i^*}\log_{d^+(i)} \frac{S_j}{S_i^*}         \\
			         & =\frac{-2}{(n+k)(n-k+1)}\left(\sum_{\substack{i \in \tilde{V}                                                                  \\d^+(i)=2}} S_i^*\sum_{j\in C(i)} \frac{S_j}{S_i^*}\log_2\frac{S_j}{S_i^*}+k\cdot k\cdot\frac{1}{k}\log_k \frac{1}{k}\right)\\
			         & = \frac{-2}{(n+k)(n-k+1)}\left(\sum_{\substack{i \in \tilde{V}                                                                 \\d^+(i)=2}} S_i\left(\frac{S_i-1}{S_i}\log_2\frac{S_i-1}{S_i}+\frac{1}{S_i}\log_2\frac{1}{S_i}\right)-k\right)\\
			         & = \frac{2}{(n+k)(n-k+1)}\left( \sum_{i=k+1}^{n} i\left( \frac{i-1}{i}\log_2\frac{i}{i-1}+\frac{1}{i}\log_2 i \right)+k \right) \\
			         & = \frac{2}{(n+k)(n-k+1)}\left(  \sum_{i=k+1}^{n}\left( (i-1)\log_2\frac{i}{i-1}+\log_2 i \right) +k \right)                    \\
			         & = \frac{2}{(n+k)(n-k+1)}\left( \log_2\frac{n^n k!}{k^k n!}+\log_2\frac{n!}{k!}+k \right)                                        \\
			         & = \frac{2}{(n+k)(n-k+1)}\left( \log_2\frac{n^n}{k^k}+k \right)                                                                 \\
			         & = \frac{2 (n \log_2 n - k \log_2 k + k)}{(n+k)(n-k+1)}.
		\end{align*}

\subsection{Expanded explanation as to why caterpillars are not always the least balanced broom trees~\label{app:catexp}}

Let $T$ be the leafy caterpillar tree with five leaves, each of size 1, and number its four internal nodes from the deepest up to the root, starting at 1. The balance score of node $i$ is then
\begin{align*}
W_i &= -\frac{i}{i+1} \log_2 \frac{i}{i+1} - \frac{1}{i+1} \log_2 \frac{1}{i+1} \\
&= \log_2 (i+1) - \frac{i}{i+1} \log_2 i.
\end{align*}
Thus $W_1 = 1$ (because node 1 is the root of the two-leaf broom head), $W_2 \approx 0.92, W_3 \approx 0.81$, and $W_4 \approx 0.72$. The weight assigned to node $i$ is $S_i = i + 1$. Hence
$$
J^1(T) = \frac{5W_4 + 4W_3 + 3W_2 + 2W_1}{5+4+3+2} \approx 0.8293.
$$
This result can be checked using the general formula for caterpillar trees with $n = 5$:
$$
J^1(T) = \frac{2 n \log_2 n}{(n-1)(n+2)} \approx 0.8293.
$$

Now suppose we modify $T$ by removing the  second lowest internal node (node 2) and reattaching its child leaf to the broom head, to create tree $T'$. Since node 2 was the second most balanced node in $T$, its removal decreases $J^1$. But the removal of node 2 also decreases the normalizing factor of $J^1$ (here equal to Sackin's index) from 14 to 12, which increases $J^1$. Nodes 3 and 4 have the same weights and balance scores as before but the weight of node 1 (the root of the broom head, which remains maximally balanced) increases from 2 to 3. Hence
$$
J^1(T') = \frac{5W_4 + 4W_3 + 3W_1}{5+4+3} \approx 0.8212.
$$

In summary, we have two factors that increase $J^1$ (decreasing the denominator and assigning more weight to $W_1=1$ in the numerator) and one that decreases $J^1$ (removing $W_2$ from the numerator). The latter effect dominates in this example and in all other cases where the leaf count is greater than 4 (Proposition~\ref{cat_prop}).

More generally, let
$$
A = \frac{1}{q_n} \left( n W_{n-1} + \dots + 4 W_3 + 3W_2 + 2W_1\right)
$$
denote the $J^1$ value of the caterpillar tree with $n > 2$ leaves (so $n-1$ internal nodes), where $q_n = n + \dots 4 + 3 + 2$, and let 
$$
B = \frac{1}{q_n - 2} \left( n W_{n-1} + \dots + 4 W_3 + 3W_1\right)
$$
denote the $J^1$ value of the broom tree created by removing the second lowest internal node and reattaching its child leaf to the broom head. Let $C = (3W_2 - W_1)/2$.
Then 
\[
B-A=  \frac{2(A-C)}{q_n-2},
\] 
where $q_n-2 > 0$.

Thus, if $A < C$ (that is, if the balance index of the caterpillar tree is less than the fixed quantity $C$) then $B < A$ (that is, the non-caterpillar broom tree is less balanced than the caterpillar). But $A$ decreases monotonically with $n$ (the caterpillar becomes less balanced as it grows). Therefore, once $A$ becomes smaller than $C$ (when $n = 4$), it remains smaller for all larger values of $n$.

\subsection{Proof of Proposition~\ref{largebrooms}}\label{largebrooms_proof}

\subsubsection{Minimum}\label{app:minimum}

We first prove the formulas for the equivalent of $\min_{2 \leq k \leq n} J_B^1(n, k, p)$ as $n \to +\infty$. 
Recall that, letting $r = k/n$: 
 \begin{equation}
 \label{eq:JB2} 
 J_B^1(n, k, p) =  2 \frac{(1-r)\log_2 n + (1- r + pr) \log_2 (1 - r + pr) - pr \log_2 pr + pr}{(1 + (2p- 1)r)(n(1- r) + 1)}.
 \end{equation}
Consider an arbitrary sequence $(k_n)$ with $2 \leq k_n \leq n$. Let $r_n = k_n/n$. Assume first that $(1-r_n) \ln n \to +\infty$. Then $(1-r_n) n\to +\infty$, and it follows from \eqref{eq:JB2} that: 
\begin{equation}
\label{eq:JB3}
J_B^1(n, k_n, p) \sim \frac{2}{1 + (2p-1) r_n} \times \frac{\log_2 n}{n}. 
\end{equation} 
In particular, if $r_n \to r < 1$ (which implies that $(1-r_n) \ln n \to +\infty$), then:
\begin{equation}
\label{eq:JB4}
J_B^1(n, k_n, p) \sim \frac{2}{1 + (2p-1) r} \times \frac{\log_2 n}{n}.
\end{equation}
From now on, let $(k_n)$ be an optimal sequence and $r_n = k_n/n$. Consider three cases: 
\textbf{Case 1: $p>1/2$.} 
It follows from \eqref{eq:JB4} that $r_n \to 1$. Therefore, due to \eqref{eq:JB2}, 
\begin{equation}\label{eq:JB_ineq}
J_B^1(n , k_n, p) \sim \frac{(1-r_n) \log_2 n + p}{p n (1- r_n + 1/n)} = \frac{1}{pn} \left[ \log_2 n + \frac{p- \frac{\log_2 n}{n}}{1-r_n + \frac{1}{n}}\right] \geq \frac{\log_2 n}{pn},
\end{equation}
where the inequality holds for $n$ large enough. Thus the smallest rate we can hope for is $\frac{\log_2 n}{pn}$. By \eqref{eq:JB3}, this is (only) achieved by any sequence $(k_n)$ such that $r_n \to 1$ and $(1-r_n) \ln n \to +\infty$, that is, that converges to $1$ sufficiently slowly. We conclude that 
\[\min_{2 \leq k \leq n} J_B^1(n, k, p) \sim \frac{\log_2 n}{pn}.\]

\textbf{Case 2: $p < 1/2$}. Assume that $r_n$ is bounded away from $1$. Then if follows from \eqref{eq:JB3} that 
\begin{equation}
\label{eq:JB5}
\min_{2 \leq k \leq n} J_B^1(n, k, p) \sim \frac{2 \log_2 n}{n}
\end{equation}
and that this is achieved if and only if $r_n \to 0$. But $r_n$ is indeed bounded away from $1$. Otherwise, a subsequence of $(r_n)$ would converge to $1$. The proof for the case $p>1/2$ shows that, along this subsequence, the best rate that could be obtained would be $\frac{\log_2 n}{pn}$. Since we now assume $p < 1/2$, this is larger than $\frac{2 \log_2 n}{n}$, hence suboptimal. This completes the proof of \eqref{eq:JB5}. The proof also shows that $r_n \to 0$. 

\textbf{Case 3:} $p = 1/2$. Due to \eqref{eq:JB3}, any sequence bounded away from $1$ leads to the rate $\frac{2 \log_2 n}{n}$. Moreover, previous arguments show that the best one can get from a sequence or subsequence going to $1$ is the same rate. Thus \eqref{eq:JB5} holds as well, and up to second order terms, all sequences $(r_n)$ except those converging very quickly to $1$ lead to the asymptotic behaviour. Note that \eqref{eq:JB5} also follows from Proposition \ref{cats_p_one_half}.

\subsubsection{Minimizers}

\textbf{Case 1: $p > \frac{1}{2}$}

        Consider the sequence of arg minimums when treating $r$ as a variable taking values on $[\frac{2}{n},1]$ (in which case, it is labeled $x$)
        \begin{align}
            x_n^*=\arg\min_{x\in [\frac{2}{n},1]}\, J_B^1(n,xn,p).
        \end{align}

        Consider the partial derivative of $f_n(x,p)=J(n,nx,p)$ with respect to $x$:
        \begin{align}\label{eq:deriv_f}
        \frac{\partial f_n(x,p)}{\partial x}=2\frac{A_p(x)n+B_p\ln n+C_p(x)+D_p(x)n\ln n}{(x(2p-1)+1)^2(n(1-x)+1)^2\ln 2},
        \end{align}
        where
        \begin{gather}
        \begin{split}
            A_p(x) =-1+2x(1-p)+(2p-1)x^2+p\ln 2-px^2\ln 2+2p^2x^2\ln 2-p\ln (px)
            \\-p(2p-1)x^2\ln (px)
            +((x-1)^2+2p^2x^2+p(x-1)(1-3x))\ln(1+x(p-1)),
        \end{split}
            \\
            B_p =-2p,\\
            C_p(x) =-1+x-2px+p\ln 2-p\ln(px)-p\ln(1+x(p-1)),\\
            D_p(x) =(1-2p)(x-1)^2.
        \end{gather}

        First, we prove that the discrete and continuous minimizers are close $|x_n^*-r_n^*|\leq \frac{1}{n}$. Using similar arguments as in Appendix~\ref{app:minimum}, $x_n^*\rightarrow 1$. Let $g_n$ denote the numerator of \ref{eq:deriv_f}. Since $A_p'(1)>0$ and $D_p'(x)>0$ on $[0,1)$ and $C_p'(x)$ is bounded on a neighborhood of $1$, it follows that there is an $\varepsilon>0$ such that on $x\in [1-\varepsilon,1]$, $A_p'(x)>a>0$ and $C_p'(x)>b$. So, we can choose $N$, depending on only on $a$ and $b$, such that $\forall n\geq N$ on $x\in [1-\varepsilon,1]$:

        \begin{align*}
            g_n'(x)=A_p'(x)n+C_p'(x)+D_p'(x)n\ln n\geq A_p'(x)n+C_p'(x)\geq an+b>0
        \end{align*}
        
        So, $g_n'(x)>0$ on $[1-\varepsilon,1]$ for all $n$ sufficiently large. Eventually, for large enough $n$, $r_n^*,x_n^*\in [1-\varepsilon,1]$. $g_n$ is thus initially negative, then zero at $x=x_n^*$, then positive. So, $x\rightarrow f_n(x,p)$ strictly decreases until $x=x_n^*$, then strictly increases (on $[1-\varepsilon,1]$). For $n\geq N$ and such that $\frac{3}{n}<\varepsilon$, it follows that $r_n^*$ must be the multiple of $1/n$ that comes immediately after or before $x_n^*$ (or both). So, $|x_n^*-r_n^*|<\frac{1}{n}$.
        
        Furthermore,
        \begin{align}\label{arg_min_condition}
            \frac{\partial f_n(x,p)}{\partial x}(x_n^*)=0\Longleftrightarrow
            0=A_p(x_n^*)n+B_p\ln n+C_p(x_n^*)+D_p(x_n^*)n\ln n.
        \end{align}

        Observe that $A_p(1)=2p^2\ln 2\neq 0$, hence in the limit with respect to $n$:

        \begin{align}
            D_p(x_n^*)n\ln n\sim -A_p(x_n^*)n,\nonumber\\
            D_p(x_n^*)\ln n\sim -(2p^2\ln 2),\nonumber\\
            (2p-1)(1-x_n^*)^2\sim \frac{2p^2}{\log_2 n},\nonumber\\
            1-x_n^*\sim \frac{p\sqrt{2}}{\sqrt{(2p-1)\log_2 n}}.
        \end{align}

   %      \iffalse
   %      Treating the $n$ as continuous, we may calculate:
   %      \begin{align}
   %          \frac{dx_n^*}{dn}=\frac{-A_p(x_n^*)-\frac{B_p}{n}-D_p(x_n^*)(\ln n+1)}{A_p'(x_n^*)n+C_p'(x_n^*)+D_p'(x_n^*)n\ln n}.
   %      \end{align}

   %      Solving for $D_p(x_n^*)$ in (\ref{arg_min_condition}) and noticing $2\sqrt{-(2p-1)D_p(x)}=D_p'(x)$ for $x\leq 1$:
   %      \begin{align}
   %          \frac{dx_n^*}{dn}=\frac{\frac{A_p(x_n^*)}{\ln n}+\frac{B_p\ln n}{n}+\frac{C_p(x_n^*)}{n}+\frac{C_p(x_n^*)}{n\ln n}}{A_p'(x_n^*)n+C_p'(x_n^*)+2\sqrt{(2p-1)\left(\frac{A_p(x_n^*)}{\ln n}+\frac{B_p\ln n}{n}+\frac{C_p(x_n^*)}{n}+\frac{C_p(x_n^*)}{n\ln n}\right)}n\ln n}.
   %      \end{align}

   %      As $A_p(x)$ is continuous over $(0,1]$, notice that $\lim_{n\rightarrow\infty}A_p(x_n^*)=A_p(1)=2p^2\ln 2$. Additionally, $A_p(x),B_p,C_p(x)$ (and their derivatives) are upper- and lower-bounded on an interval $[x_0,1]$, meaning that asymptotically:
   %      \begin{align*}
			% \frac{dx_n^*}{dn}\stackrel{n\to\infty}{\sim} & \frac{\frac{A_p(x_n^*)}{\ln n}}{2n\sqrt{(2p-1)A_p(x_n^*)\ln n}} \\
			% \stackrel{n\to\infty}{\sim}                                        & \frac{p\sqrt{\ln 2}}{n(\ln n)^{\frac{3}{2}}\sqrt{2(2p-1)}}.
   %      \end{align*}
        
   %      Now
   %      \begin{align*}
   %          \lim_{n\rightarrow\infty}\frac{1-x_n^*}{\frac{\sqrt{2}p}{\sqrt{(2p-1)\log_2 n}}}=\lim_{n\rightarrow\infty}\frac{-\frac{dx_n^*}{dn}}{-\frac{(\sqrt{\ln 2})p}{\sqrt{2(2p-1)}n(\ln n)^{\frac{3}{2}}}}=1.
   %      \end{align*}
   %      \fi

        Using $|x_n^*-r_n^*|\leq \frac{1}{n}$, the above result, and the Squeeze Theorem, the proof is complete.
        
\textbf{Case 2: $p< \frac{1}{2}$}

For all $x\in [0,\frac{1}{2}]$, $nD_p(x)-B_p\geq \frac{n(1-2p)}{4}-2p>0$ for $n$ sufficiently large. Moreover, $A_p(x),C_p(x)\rightarrow\infty$ as $x\rightarrow 0^+$. Hence by \ref{eq:deriv_f}, there is a $\varepsilon>0$ such that for all sufficiently large $n$, $x\rightarrow f_n(x,p)$ is strictly increasing on $(0,\varepsilon]\supseteq [\frac{2}{n},\varepsilon]$. From Appendix~\ref{app:minimum}, $r_n^* \to 0$ for $p<\frac{1}{2}$. Eventually then, for all $n$ large enough, $r_n^*\in [\frac{2}{n},\varepsilon]$. But because $x\rightarrow f_n(x,p)$ is strictly increasing within this interval, $r_n^*=\frac{2}{n}$, or equivalently $k_n^*=2$, for all sufficiently large $n$.

\textbf{Case 3: $p=\frac{1}{2}$}

This follows from Proposition \ref{cats_p_one_half}

\subsection{Proof of Proposition~\ref{cats_p_one_half}}\label{cats_p_one_half_proof}

\begin{proof}
\begin{multline*}
\frac{\partial J_B^1}{\partial k}\left(n, k, \frac{1}{2}\right) = \\ 
\frac{-(n+1)\ln k + (n-1)\ln(2n-k) - 2(n+1-k) + (n + 3)\ln 2}{n (n + 1 - k)^2 \ln 2} \\
= \frac{n \ln \left( \frac{2n}{k} - 1 \right) - \ln(k(2n-k)) - 2(n+1-k) + (n + 3)\ln 2}{n (n + 1 - k)^2 \ln 2}.
\end{multline*}
We then note that,  for all $2 \le k \le n$,
$$
n \ln \left( \frac{2n}{k} - 1 \right) = 
2 \left( n-k + \frac{(n-k)^3}{3n^2} + \frac{(n-k)^5}{5n^4} + \dots \right) \ge 2(n-k)
$$
and $\ln(k(2n-k)) \le 2 \ln n$. Hence
$$
\frac{\partial J_B^1}{\partial k}\left(n, k, \frac{1}{2}\right) \ge 
\frac{(n + 3)\ln 2 - 2(\ln n + 1)}{n (n + 1 - k)^2 \ln 2},
$$
which is positive for all $n \ge 4$. Therefore $J_B^1(n, k, \frac{1}{2})$ increases with $k$, which implies it is minimal when $k = 2$.
\end{proof}

\subsection{Proof of Proposition~\ref{asymptotic}}\label{asym_JB_p}

    \begin{proof}
    We first reparameterize $J_B^1$ as
        \[J_B^1(n, rn, p) = \frac{2}{n} \times \frac{(1-r)\log_2 n + (1- r + pr) \log_2 (1 - r + pr) - pr \log_2 pr + pr}{(1 + (2p- 1)r)( 1- r + 1/n)}.\]

    \subsubsection*{As $p\rightarrow\infty$}

    Let $n$ be fixed throughout. Firstly note that for fixed $n$ and $k$,  as $p\to \infty$, 
    %\begin{align*}
\[pr\log_2(1-r+pr)-pr\log_2(pr) = 
    \frac{pr}{\ln 2} \ln \left(1+\frac{1-r}{pr}\right) \sim \frac{pr}{\ln 2} \times \frac{1-r}{pr} = \frac{1 -r}{\ln 2} 
    %\\
    %= \lim_{p\rightarrow\infty}\frac{-\frac{1}{\ln 2(1+(1-r)/(pr))}\frac{1-r}{p^{2}r}}{-\frac{1}{p^{2}r}}=\frac{1-r}{\ln 2}.
    \]
    %\end{align*}
    %\begin{align*}
    %\lim_{p\rightarrow\infty}pr\log_2(1-r+pr)-pr\log_2(pr) &= 
    %\lim_{p\rightarrow\infty}\frac{\log_2(1+(1-%r)/(pr))}{(pr)^{-1}} \\
    %&= \lim_{p\rightarrow\infty}\frac{-\frac{1}{\ln 2(1+(1-r)/(pr))}\frac{1-r}{p^{2}r}}{-\frac{1}{p^{2}r}}=\frac{1-r}{\ln 2}.
    %\end{align*}

    Thus, the leading term is $2pr$ in the numerator and $2pr(n-nr+1)$ in the denominator. Therefore, $(p\to J_B^1(n,k,p))\sim \frac{1}{n-nr+1}$. For fixed $n$ and large enough $p$, $r^*=\frac{2}{n}$ (a constant with respect to $p$).
    \iffalse\footnote{\yan{
    %Two writing problems here. First, Equation \eqref{eq:deriv_f} comes later, and we should try to avoid referring to something that comes later. Second, 
    The expression ``for $p$ large enough" is not completely clear here. Does the threshold $P$ depend on $n$ or not? Do we mean: $\exists P, \forall p \geq P, \forall n, r^* = 2/n$, or $\forall n, \exists P, \forall p \geq P, r^* = 2/n$. Note also that I tend to use the notation $r_n^*$, and note $r^*$, to make explicit that the optimal $r$ depends on $n$. Also, try to break long sentences into several shorter ones. This is usually easier to follow.} \arm{Indeed it depends on $n$}}\fi
    This may be seen by using \eqref{eq:deriv_f} and noticing that the leading term in $p$ in the numerator is $4p^2x^2\ln(1+x(p-1))n$ for any $x\in \left[\frac{2}{n},1\right]$. It follows that there exists $P$ such that, for any $p>P$, $f_n$ is increasing in $x$ over this interval, hence so is $J_B^1$ in terms of $r\in\{\frac{2}{n},\frac{3}{n},\dots,\frac{n}{n}\}$. Therefore, $J_B^1(n,k^*,p)\sim \frac{1}{n-1}$.

    \subsubsection*{As $p\rightarrow 0$}

    Let $n$ be fixed. We first show that for $p$ small enough, %sufficiently small $p$ ($p<\varepsilon$, for some $\varepsilon>0$), 
    we have $k^*=n-1$. Notice that:

    \[\lim_{p\rightarrow 0^+} J_B^1(n,k,p)=\frac{2\log_2(n-k)}{n-k+1}\equiv J_B^1(n,k,0)\]

    for $k\in \{2,\dots,n-1\}$ and for $k=n$, $\lim_{p\rightarrow 0^+} J_B^1(n,k,p)=1$. We find the critical points of $J_B^1(n,k,0)$ at fixed $n$, which becomes

    \[2\left[\frac{\frac{1}{(n-k)\ln 2}}{n-k+1}\cdot (-1)+\frac{\log_2(n-k)}{(n-k+1)^2}\right]=0\]

    \[\Leftrightarrow n-k+1=(n-k)\ln (n-k)\Leftrightarrow 1=(\ln (n-k)-1)e^{\ln(n-k)}\]

    \[\Leftrightarrow k=n-e^{W_0(e^{-1})+1}\]
    where $W_0$ is the real, increasing branch of the Lambert W Function. This is on the interior of $[2,n-1]$ for $n\geq 6$, is a unique critical point, and is the maximum of $J_B^1(n,k,0)$ for fixed $n$. Thus, the global minimum must be on $2$ or $n-1$. By direct computation $0=J_B^1(n,n-1,0)<J_B^1(n,2,0)$, proving $k^*=n-1$.
    
    So,
    \[J_B^1(n,k^*,p)=\log_2 n+(1+p(n-1))\log_2(\frac{1}{n}+p\frac{n-1}{n})-p(n-1)\log_2(p\frac{n-1}{n})+p(n-1).\]

    As $x\rightarrow 0^+$, terms of the forms $t(x)=x\ln (1+x),t(x)=\ln (1+x),t(x)=x$ decay much faster to $0$ compared to $x\ln x$ (meaning $\lim_{x\rightarrow 0^+}\frac{t(x)}{x\ln x}=0$), so asymptotically:

    \[J_B^1(n,k^*,p)\sim p(1-n)\log_2(p).\]
        
    \end{proof}

\subsection{Proof of Proposition \ref{R3_R4_boundary}}\label{R3_R4_boundary_proof}

\subsubsection*{$R_3$-$R_4$ boundary}

%Here I give a rough large $n$ formula for the boundary between these two regions. 
The approximate condition to be on the boundary is $J_B^1(n,2,p)=J_B^1(n,n-1,p)$ (approximate, as $n$ is assumed to be continuous), which is equivalent to
% \begin{multline*}
% \frac{2[2p+(2p+n-2)\log_2 (2p+n-2) -2p\log_2(2p)]}{(4p+n-2)(n-1)} \\
% =\frac{2[(n-1)p+((n-1)p+n-(n-1))\log_2((n-1)p+n-(n-1)) -(n-1)p\log_2 (n-1)p]}{(2(n-1)p+n-(n-1))(n-(n-1)+1)},
% \end{multline*}
\begin{multline*}
\frac{2[(2p+n-2)\log_2 (2p+n-2) -2p\log_2(p)]}{(4p+n-2)(n-1)} \\
=\frac{[(n-1)p+((n-1)p+1)\log_2((n-1)p+1) -(n-1)p\log_2 (n-1)p]}{2(n-1)p+1}.
\end{multline*}

Note that one key fact about the boundary is that as $n\rightarrow\infty$, $p\rightarrow 0^+$. Notice that we have the following possible options for the asymptotic behavior of $(n,p)$:

\begin{enumerate}
    \item As $n\rightarrow+\infty$, $p$ is bounded below
    \item As $p\rightarrow 0^+$, $n$ is bounded above
    \item As $n\rightarrow+\infty$, $p\rightarrow 0^+$
\end{enumerate}

The first and second option are precluded by Proposition \ref{largebrooms} and Proposition \ref{asymptotic}. Now then:

%Suppose instead that there is a $\varepsilon>0$ such that for all $N\in\mathbb{N}$, there is a $n>N$ such that $p\geq \varepsilon$. Consider a $0<p'<\varepsilon$, where $p'$ is small enough such that for a given $n$, $k^*=n-1$, or equivalently $(n,p')\in R_4$. Then asymptotically, there will always be a $n$ larger than any given $N$ such that $p'<p$, meaning $\arg\min_k J_B^1(n,k,p')=n-1\neq 2$, which implies $\lim_{n\rightarrow\infty}\arg\min_k J_B^1(n,k,p)\neq 2$, a contradiction. This leads to:
%Using this small $p$, large $n$ approximation, the above roughly becomes
\begin{gather*}
\frac{2\log_2 n}{n}
=\frac{[np+(np+1)\log_2(np+1) -np\log_2 np]}{2np+1} \\
\Leftrightarrow\quad
\frac{2(2np+1)\log_2 n}{n}-np-(np+1)\log_2(np+1) +np\log_2 np=0.
\end{gather*}

Since $np\sim np+p\log_2 n$, we can further simplify:
$$
-\frac{2\log_2 n}{n}+np+(np+1)\log_2(np+1) -np\log_2 np=0,
$$
and then
\begin{equation}
\label{eq:int_bound}
    np\log_2(np+1) -np\log_2 np
    =-\log_2(np+1)+\frac{2\log_2 n}{n}-np.
\end{equation}

Notice the left-hand side is always positive, meaning
\begin{align*}
&-\log_2(np+1)+\frac{2\log_2 n}{n}-np>0 \\
\Leftrightarrow&\quad
n^{2/n}>(np+1)e^{np\ln 2} \\
\Leftrightarrow&\quad
n^{2/n}\ln 4>(np\ln 2+\ln 2)e^{np\ln 2+\ln 2} \\
\Leftrightarrow&\quad
W_0(n^{2/n}\ln 4)>np\ln 2+\ln 2 \\
\Leftrightarrow&\quad
\frac{W_0(e^{2\ln n/n}\ln 4)-\ln 2}{n\ln 2}>p,
\end{align*}
where $W_0$ denotes the real, increasing branch of the Lambert W function.

Note that as $n\rightarrow\infty$, $W_0(e^{2\ln n/n}\ln 4)\rightarrow W_0(\ln 4)$. Additionally, a point on the graph of $xe^x=y$ is $(\ln 2,\ln 4)$, meaning $W_0(\ln 4)=\ln 2$. Finally, $W_0'(x)=\frac{W_0(x)}{x(W_0(x)+1)}$, so $W_0'(\ln 4)=\frac{\ln 2}{\ln 4(\ln 2+1)}=\frac{1}{2(\ln 2+1)}$. Using the first order approximation of $W_0(x)$ centered at $\ln 4$ and $e^x$ centered at $0$ ($\ln 2+\frac{1}{2(\ln 2+1)}(x-\ln 4)$ and $1+x$, respectively), we obtain the asymptotic result:
\begin{gather*}
\frac{\ln 2+\frac{1}{2(\ln 2+1)}((1+2\ln n/n)\ln 4-\ln 4)-\ln 2}{n\ln 2}>p \\
\Leftrightarrow\quad
\frac{\frac{1}{(\ln 2+1)}(2\ln n/n)}{n}>p 
\quad\Leftrightarrow\quad
\frac{2\ln n}{(\ln 2+1)n^2}>p.
\end{gather*}

Note further that asymptotically $\frac{2\ln n}{(\ln 2+1)n}>np$, so namely $np\rightarrow 0$. Furthermore, in the limit $np\rightarrow 0$, $np\log_2(np)$ goes to $0$ much slower than any other $np$ term in \ref{eq:int_bound}, so we reduce \ref{eq:int_bound} to:
\begin{gather*}
np\log_2 np=-\frac{2\log_2 n}{n} 
\quad\Leftrightarrow\quad
(np)^{np}=n^{-\frac{2}{n}}.
\end{gather*}

Notice the above becomes:
\begin{gather*}
\ln(np)e^{\ln (np)}=-\frac{2}{n}\ln n
\quad\Leftrightarrow\quad
\ln (np)=W_{-1}(-\frac{2}{n}\ln n)
\quad\Leftrightarrow\quad
p=\frac{-\frac{2}{n}\ln n}{nW_{-1}(-\frac{2}{n}\ln n)},
\end{gather*}
where $W_{-1}(x)$ is the decreasing, real branch of the Lambert W function. As $x\rightarrow 0^-$, $W_{-1}(x)\rightarrow-\infty$ and $W_{-1}(x) \sim \ln(-x)$, so asymptotically we obtain
$$
			     p \stackrel{n\to\infty}{\sim}  \frac{-\frac{2}{n}\ln n}{n\ln(\frac{2}{n}\ln n)} 
			=\frac{-\frac{2}{n}\ln n}{n(\ln 2-\ln n+\ln\ln n)}
			\stackrel{n\to\infty}{\sim} \frac{2}{n^2}.
$$

\subsubsection*{$R_1$-$R_2$ boundary}

A similar analysis to that presented above as $p \to \infty$ and $n \to \infty$ shows that the $R_1$-$R_2$ boundary approaches $p = n^2 \log_2 n / 3$.

\subsection{Proof of Proposition~\ref{min_J1_in_p}}\label{min_J1_in_p_proof}

\begin{proof}
    Since $J_B^1$ is a continuous function of $p$, $\min_{2 \le k \le n} J^1_B(n,k,p)$ is also a continuous function of $p$, so it suffices to prove the result separately for each region of $n$-$p$ space.

    \subsubsection*{Regions $R_1$ and $R_3$}
    
    By Proposition~\ref{broom_J1_in_p}, the result holds in the regions $R_1$ and $R_3$ where $k^* = 2$. 
    
    \subsubsection*{Region $R_2$}
    
    If $p > \frac{1}{2}$ then $\theta(k, p) < 4/k + k/2 < k$ for all $k \ge 3$. We then have $n \ge k > \theta(k, p)$ for all $n \ge 3$, $k \in \{3, \dots, n \}$ and $p > \frac{1}{2}$, which, by Proposition~\ref{cats_p_one_half}, includes the region $R_2$. Therefore, by Proposition~\ref{broom_J1_in_p}, $J^1_B(n,k,p)$ is a strictly decreasing function of $p$ for all $(n,p) \in R_2$ and all $k \in \{2, \dots, n-1 \}$. It follows that in $R_2$, for any $\epsilon > 0$,
    \begin{align*}
    \min_{2 \le k \le n} J^1_B(n,k,p + \epsilon) &= \min_{3 \le k \le n-2} J^1_B(n,k,p + \epsilon) \\
    &= \min \{ J^1_B(3,k,p + \epsilon), \dots, J^1_B(n,n-2,p + \epsilon) \} \\
    &< \min \{ J^1_B(3,k,p), \dots, J^1_B(n,n-2,p) \} \\
    &= \min_{3 \le k \le n-2} J^1_B(n,k,p) \\
    &= \min_{2 \le k \le n-1} J^1_B(n,k,p).
    \end{align*}
    Hence $\min_{2 \le k \le n} J^1_B(n,k,p)$ is a strictly decreasing function of $p$ for all $(n,p) \in R_2$. Lastly, we have $\theta(2, p) = 1/p + 2(1-p) \le 3$ for all $p \ge \frac{1}{2}$, which implies $n > \theta(2, p)$ for all $(n,p) \in R_2$, as required.
    
    \subsubsection*{Region $R_4$}
    
    First we note that 
    \begin{equation}\label{equivaleneces}
    n < \theta(n - 1, p) \Leftrightarrow p < \frac{1}{n-1} \Leftrightarrow n < \frac{p+1}{p}.
    \end{equation}
    Since $J^1_B(n, n-1, 1/(n-1)) = 1$, we cannot have $k^* = n - 1$ along the curve $p =1/(n-1)$, so the curve must lie entirely above $R_4$. Therefore $n < \theta(n - 1, p)$ by Equation~\ref{equivaleneces}, and Proposition~\ref{broom_J1_in_p} then implies that $J^1_B(n,n-1,p)$ is a strictly increasing function of $p$ for all $(n, p) \in R_4$. 
    By Equation~\ref{equivaleneces}, $n < \theta(n - 1, p)$ also implies 
    $$
    n < \frac{p+1}{p} < \frac{1}{p} + 2(1-p) = \theta(2, p)
    $$ 
    for all $p < 1$, which, by Proposition~\ref{cats_p_one_half}, includes all $p$ values in $R_4$. Hence, for all $(n,p) \in R_4$, we have shown first that $\min_{2 \le k \le n} J^1_B(n,k,p)$ is a strictly increasing function of $p$ and second that $n < \theta(2, p)$, which proves the result for $R_4$.
\end{proof}

\subsection{Proof of Corollary~\ref{max_min_J1B}}\label{max_min_J1B_proof}

\begin{proof}
    In the proof of Proposition~\ref{min_J1_in_p}, we showed that $n > \theta(2, p)$ for all $n > 3$ and $p > \frac{1}{2}$, which includes all of $R_1$ and $R_2$. We also showed that $n < \theta(2, p)$ for all $(n,p) \in R_4$. Therefore $n = \theta(2, p)$ must hold only in $R_3$, where $k^* = 2$. Hence when $\min_{2 \le k \le n} J^1_B(n,k,p)$ is maximal, $k^* = 2$ and $n = \theta(2, p)$, which has the unique positive solution 
    $$
    p = \frac{1}{4} \left( \sqrt{n^2-4n+12} - n + 2 \right).
    $$
    Substituting this $p$ value into $J^1_B(n, 2, p)$ gives the required result.
\end{proof}

\newpage

\begin{table}[h]
	\centering
	\renewcommand{\arraystretch}{2.5}
	\begin{tabular}{|c|c|c|c|c|c|}
		\hline
  $n$ 
  & $\mathbb{H}_Y(I_S)$ 
  & $\mathbb{E}_Y(I_S)$ 
  & \thead{$\mathbb{E}_Y(J^1) =$ \\ $\dfrac{n\log_2n}{\mathbb{H}_Y(I_S)}$} 
  & \thead{$\mathbb{H}_Y(J^1) =$ \\
  $\dfrac{n\log_2n}{\mathbb{E}_Y(I_S)}$} 
  & \thead{$\mathcal{J}(n) = $ \\ 
   $\mathbb{E}_Y(J^1) - \mathbb{H}_Y(J^1)$}\\
		\hline 
  $2$ & $2$ & $2$ & 1 & 1 & 0 \\
		\hline 
  $3$ & $5$ & $5$ & 0.9509775 & 0.9509775 & 0 \\
		\hline 
  $4$ & $\dfrac{216}{25}$ & $\dfrac{26}{3}$ & $0.\overline{925}$ & $0.\overline{923076}$ & $0.\overline{002849}$ \\
		\hline 
  $5$ & $\dfrac{728}{57}$ & $\dfrac{77}{6}$ & 0.9089966 & 0.9046473 & 0.0043493 \\
		\hline 
  $6$ & $\dfrac{1162800}{67217}$ & $\dfrac{87}{5}$ & 0.8965605 & 0.8913664 & 0.0051941 \\
		\hline 
  $7$ & $\dfrac{199806750}{9017743}$ & $\dfrac{223}{10}$ & 0.8869172 & 0.8812325 & 0.0056847 \\
		\hline 
  $8$ & $\dfrac{14827566600}{543154091}$ & $\dfrac{962}{35}$ & 0.8791529 & 0.8731809 & 0.0059720 \\
		\hline 
  $9$ & $\dfrac{3276378490018233600}{100225924446507833}$ & $\dfrac{4609}{140}$ & 0.8727252 & 0.8665883 & 0.0061369 \\
		\hline 
  $10$ & $\dfrac{247828225931780785200}{6470294678155594501}$ & $\dfrac{4861}{126}$ & 0.8672884 & 0.8610634 & 0.0062249 \\
		\hline
  $11$ & $\dfrac{4355154127379809158517035000}{98723552636267773621908727}$ & $\dfrac{55991}{1260}$ & 0.8626104 & 0.8563470 & 0.0062634 \\
		\hline
	\end{tabular}
	\caption{Harmonic and arithmetic means of $I_S$ and $J^1$ under the Yule model, and corresponding values of the Jensen gap $\mathcal{J}(n)$. Non-repeating decimal numbers are rounded to seven decimal places. $H_Y(I_S)$ values are given as exact rational numbers to aid the search for a better general approximation or an exact formula.}
	\label{table_values}
\end{table}

\begin{table}[h]
	\centering
	\renewcommand{\arraystretch}{2.5}
	\begin{tabular}{|c|c|c|c|c|c|}
		\hline
  $n$ 
  & $\mathbb{H}_U(I_S)$ 
  & $\mathbb{E}_U(I_S)$ 
  & \thead{$\mathbb{E}_U(J^1) =$ \\ $\dfrac{n\log_2n}{\mathbb{H}_U(I_S)}$} 
  & \thead{$\mathbb{H}_U(J^1) =$ \\
  $\dfrac{n\log_2n}{\mathbb{E}_U(I_S)}$} 
  & \thead{$\mathcal{J}(n) = $ \\ 
   $\mathbb{E}_U(J^1) - \mathbb{H}_U(J^1)$}\\
		\hline 
  $2$ & $2$ & $2$ & 1 & 1 & 0 \\
		\hline 
  $3$ & $5$ & $5$ & 0.9509775 & 0.9509775 & 0 \\
		\hline 
  $4$ & $\dfrac{360}{41}$ & $\dfrac{44}{5}$ & $0.9\overline{1}$ & $0.\overline{90}$ & $0.0\overline{02}$ \\
		\hline 
  $5$ & $\dfrac{3822}{289}$ & $\dfrac{93}{7}$ & 0.8778614 & 0.8738439 & 0.0040174 \\
		\hline 
  $6$ & $\dfrac{4883760}{267509}$ & $\dfrac{386}{21}$ & 0.8495512 & 0.8437960 & 0.0057552 \\
		\hline 
  $7$ & $\dfrac{2051349300}{86119541}$ & $\dfrac{793}{33}$ & 0.8250067 & 0.8177793 & 0.0072274 \\
		\hline 
  $8$ & $\dfrac{11813334132600}{395454217009}$ & $\dfrac{12952}{429}$ & 0.8034058 & 0.7949351 & 0.0084707 \\
		\hline 
  $9$ & $\dfrac{2440219396211496900}{67072134682154831}$ & $\dfrac{26333}{715}$ & 0.7841601 & 0.7746351 & 0.0095250 \\
		\hline 
  $10$ & $\dfrac{3265803434049000144300}{75388088210344282097}$ & $\dfrac{106762}{2431}$ & 0.7668367 & 0.7564121 & 0.0104246 \\
		\hline
	\end{tabular}
	\caption{Harmonic and arithmetic means of $I_S$ and $J^1$ under the uniform model, and corresponding values of the Jensen gap $\mathcal{J}(n)$. Non-repeating decimal numbers are rounded to seven decimal places. $H_U(I_S)$ values are given as exact rational numbers to aid the search for a better general approximation or an exact formula.}
	\label{table_values_unif}
\end{table}
\end{appendices}
	
\end{document}